\documentclass[12pt,letterpaper]{article}
\usepackage{amsmath,amsfonts,amsthm,amssymb}

\swapnumbers 

\newtheorem{lemma}{Lemma}[section]
\newtheorem{corollary}[lemma]{Corollary}
\newtheorem{theorem}[lemma]{Theorem}
\newtheorem{assumptions}[lemma]{Standing Assumptions}
\newtheorem{assumption}[lemma]{Standing Assumption}

\theoremstyle{definition} 
\newtheorem{definition}[lemma]{Definition}
\newtheorem{remark}[lemma]{Remark}

\newtheorem{remarks}[lemma]{Remarks}

\newcommand{\Nat}{{\mathbb N}}

\newcommand\reals{{\mathbb R}}

\newcommand{\dg}{\sp{\text{\rm o}}}
\newcommand\bul{\noindent$\bullet\ $}
\newcommand{\sgn}{\operatorname{sgn}}

\begin{document}

\title{Symmetries in Synaptic Algebras}

\author{David J. Foulis{\footnote{Emeritus Professor, Department of
Mathematics and Statistics, University of Massachusetts, Amherst,
MA; Postal Address: 1 Sutton Court, Amherst, MA 01002, USA;
foulis@math.umass.edu.}}\hspace{.05 in} and Sylvia
Pulmannov\'{a}{\footnote{ Mathematical Institute, Slovak Academy of
Sciences, \v Stef\'anikova 49, SK-814 73 Bratislava, Slovakia;
pulmann@mat.savba.sk. The second author was supported by Research and 
Development Support Agency under the contract No. APVV-0178-11 and grant 
VEGA 2/0059/12.}}}
\date{}

\maketitle

\begin{abstract}

\noindent A synaptic algebra is a generalization of the Jordan 
algebra of self-adjoint elements of a von Neumann algebra. We 
study symmetries in synaptic algebras, i.e., elements whose 
square is the unit element, and we investigate the equivalence 
relation on the projection lattice of the algebra induced by 
finite sequences of symmetries. In case the projection lattice 
is complete, or even centrally orthocomplete, this equivalence 
relation is shown to possess many of the properties of a dimension 
equivalence relation on an orthomodular lattice.
\end{abstract}

\section{Introduction} \label{sc:intro} 

Synaptic algebras, which were introduced in \cite{FSynap} and further 
studied in \cite{FPSynap, PuNote} tie together the notions of an 
order-unit normed space \cite[p. 69]{Alf}, a special Jordan algebra 
\cite{McC}, a convex effect algebra \cite{GPBB}, and an orthomodular 
lattice \cite{Beran, Kalm}. The self-adjoint part of a von Neumann 
algebra is an example of a synaptic algebra; see \cite{FSynap, FPSynap, 
PuNote} for numerous additional examples. 

Our purpose in this article is to study symmetries $s$ in a synaptic 
algebra $A$ and the equivalence relation $\sim$ induced by finite 
sequences of symmetries on the orthomodular lattice $P$ of all 
projections $p$ in $A$. For a symmetry $s$, we have $s\sp{2}=1$ 
(the unit element in $A$), and $p\sp{2}=p$ for a projection $p$. 
If $P$ is a complete lattice, or even centrally orthocomplete, i.e., 
every family of projections that is dominated by an orthogonal family 
of central projections has a supremum, then we show that $\sim$ 
acquires many of the properties of a dimension equivalence relation 
on an orthomodular lattice \cite{Loom}.  

In Section \ref{sc:BPSA} we review the definition and basic 
properties of a synaptic algebra $A$. Since the projections in 
$A$ form an orthomodular lattice (OML) $P$, we sketch in Section 
\ref{sc:OML} a portion of the theory of OMLs that will be needed 
for our subsequent work.  In Section \ref{sc:latproj} we focus on 
the special properties of the OML $P$ that are acquired due to 
the fact that $P\subseteq A$.  In Section \ref{sc:sym&persp}  
we introduce the notion of a symmetry $s$ in $A$, study 
exchangeability of projections by a symmetry, and relate 
symmetries to the notion of perspectivity of projections. The 
condition of central orthocompleteness is defined in Section 
\ref{sc:centOC}, and it is observed that, if $P$ is centrally 
orthocomplete, then the center of $A$ is a complete boolean 
algebra and $A$ hosts a central cover mapping. From Section 
\ref{sc:centOC} onward, it is assumed that $P$ is, at least, 
centrally orthocomplete. The equivalence relation $\sim$ on $P$ 
induced by finite sequences of symmetries is introduced in 
Section \ref{sc:equivalenceproj} where we investigate the 
extent to which $\sim$ is a dimension equivalence relation. 
Finally, in Section \ref{sc:CompleteP} we cover some of the 
features of the relation of exchangeability of projections by 
symmetries that require completeness of the OML $P$.    

\section{Basic Properties of a Synaptic Algebra} \label{sc:BPSA}

In this section, we review the definition of a synaptic algebra and 
sketch some basic facts about synaptic algebras. For more details, 
see \cite{FSynap, FPSynap, PuNote}. We use the notation $:=$ for 
``equals by definition" and ``iff" abbreviates ``if and only if."

\begin{definition}[{\cite[Definition 1.1]{FSynap}}] 
\label{df:SynapticAlgebra}
Let $R$ be a linear associative algebra with unity element $1$ 
over the real numbers $\reals$, and let $A$ be a real vector 
subspace of $R$. Let $a,b\in A$. We understand that the product 
$ab$ is calculated in $R$, and that it may or may not belong to $A$.  
We write $aCb$ iff $a$ and $b$ commute (i.e. $ab=ba$) and we define 
$C(a):=\{b\in A:aCb\}$. If $B\subseteq A$, then $C(B):=\bigcap
\sb{b\in B}C(b)$, $CC(B):=C(C(B))$, and $CC(b):=C(C(b))$.
 
The vector space $A$ is a \emph{synaptic algebra} with 
\emph{enveloping algebra} $R$ iff the following conditions are 
satisfied:  
\begin{enumerate}
\item[SA1.] $A$ is a partially ordered archimedean real vector space 
 with positive cone $A\sp{+}=\{a\in A:0\leq a\}$, $1\in A\sp{+}$ is 
 an order unit in $A$, and $\|\cdot\|$ is the corresponding order 
 unit norm on $A$. 
\item[SA2.] If $a\in A$ then $a\sp{2}\in A\sp{+}$.
\item[SA3.] If $a,b\in A\sp{+}$, then $aba\in A\sp{+}$.
\item[SA4.] If $a\in A$ and $b\in A\sp{+}$, then $aba=0\Rightarrow 
 ab=ba=0$.
\item[SA5.] If $a\in A\sp{+}$, there exists $b\in A\sp{+}\cap CC(a)$ 
 such that $b\sp{2}=a$.
\item[SA6.] If $a\in A$, there exists $p\in A$ such that $p=p\sp{2}$ and, 
 for all $b\in A$, $ab=0\Leftrightarrow pb=0$.
\item[SA7.] If $1\leq a\in A$, there exists $b\in A$ such that $ab=ba=1$.
\item[SA8.] If $a,b\in A$, $a\sb{1}\leq a\sb{2}\leq a\sb{3}\leq\cdots$ 
 is an  ascending sequence of pairwise commuting elements of $C(b)$ 
 and $\lim\sb{n\rightarrow\infty}\|a-a\sb{n}\|=0$, then $a\in C(b)$.
\end{enumerate} 
We define $P :=\{p\in A:p=p\sp{2}\}$ and we refer to elements $p
\in P$ as \emph{projections}. Elements $e$ in the ``unit interval" 
$E:=\{e\in A:0\leq e\leq 1\}$ are called \emph{effects}. The set $C(A)$ 
is called the \emph{center} of $A$. We understand that subsets of $A$ 
such as $P$, $E$, and $C(A)$ are partially ordered by the respective 
restrictions of the partial order $\leq$ on $A$. If $p,q\in P$ and 
$p\leq q$, we say that $p$ is a \emph{subprojection} of $q$, or 
equivalently, that $q$ \emph{dominates} $p$.
\end{definition}

\begin{assumptions}
For the remainder of this article, $A$ is a synaptic algebra  with 
unit $1$, with enveloping algebra $R$, with $E$ as its unit interval, 
and with $P$ as its set of projections. To avoid trivialities, we shall 
assume that $A$ is ``non-degenerate," i.e., $0\not= 1$. Also, we shall 
follow the usual convention of identifying each real number $\lambda
\in\reals$ with the element $\lambda1\in A$, so that $\reals\subseteq C(A)$.
\end{assumptions}

As $A$ is an order unit space with order unit $1$, the
\emph{order-unit norm} $\|\cdot\|$ is defined on $A$ by $\|a\|
:=\inf\{0<\lambda\in\reals:-\lambda\leq a\leq\lambda\}$. 
If $a\in A$, then by \cite[Theorem 8.11]{FSynap}, $C(a)$ is
norm closed in $A$. In fact, it can be shown that, in the 
presence of axioms SA1--SA7, axiom SA8 is equivalent to the 
condition that $C(a)$ is norm closed in $A$ for all $a\in A$.

Since $A$ is closed under squaring, it is a \emph{special Jordan 
algebra} under the Jordan product 
\[
a\circ b:=\frac12((a+b)\sp{2}-a\sp{2}-b\sp{2})=\frac12(ab+ba)\in A
 \text{\ for all\ }a,b\in A. 
\]
If $a,b\in A$, then $ab+ba=2(a\circ b)\in A$ and $aCb\Rightarrow 
ab=ba=a\circ b=b\circ a\in A$. Also, $aba=2a\circ(a\circ b)-a\sp{2}
\circ b\in A$.

\begin{definition} [{\cite[Definition 4.1]{FSynap}}] \label{df:quadmap}
If $a\in A$, the mapping $J\sb{a}\colon A\to A$ defined for $b\in A$ 
by $J\sb{a}(b):=aba$ is called the \emph{quadratic mapping} determined 
by $a$. If $p\in P$, then the quadratic mapping $J\sb{p}$ is called the 
\emph{compression} determined by $p$ \cite{FCPOAG}.
\end{definition}
\noindent If $a\in A$, than by \cite[Theorem 4.2 and Lemma 4.4]
{FSynap} the quadratic mapping $J\sb{a}\colon A\to A$ is linear, order 
preserving, and norm-continuous. In particular, if $0\leq b\in A$, 
then $0\leq J\sb{a}(b)=aba$, which is a stronger version of axiom SA3. 

If $a,b,c\in A$, then $abc$ belongs to $R$, but not necessarily 
to $A$. However, we have the following.

\begin{lemma} \label{le:abc}
If $a,b,c\in A$, then $abc+cba\in A$.
\end{lemma}

\begin{proof} $abc+cba=(a+c)b(a+c)-aba-cbc=J\sb{a+c}(b)-J\sb{a}(b)
-J\sb{c}(b)\in A$. 
\end{proof}

By \cite[Theorem 2.2]{FSynap}, each element $a\in A\sp{+}$ has
a uniquely determined \emph{square root} $a\sp{1/2}\in A\sp{+}$
such that $(a\sp{1/2})\sp{2}=a$\,; moreover, $a\sp{1/2}\in CC(a)$.
If $a\in A$, then $a\sp{2}\in A\sp{+}$, whence $a$ has an
\emph{absolute value} $|a| :=(a\sp{2})\sp{1/2}\in CC(a\sp{2})
\subseteq CC(a)$ which is uniquely determined by the properties
$|a|\in A\sp{+}$ and $|a|\sp{2}=a\sp{2}$.

By \cite[Lemma 7.1 and Theorem 7.2]{FSynap}, an element $a\in A$
has an \emph{inverse} $a\sp{-1}\in A$ such that $aa\sp{-1}=a\sp{-1}a
=1$ iff there exists $0<\epsilon\in\reals$ such that $\epsilon\leq
|a|$; moreover, if $a$ is \emph{invertible} (i.e., $a\sp{-1}$
exists in $A$), then $a\sp{-1}\in CC(a)$.

If $a\in A$, then by \cite[Theorem 3.3]{FSynap}, 
\[
a\sp{+}:=\frac12(|a|+a)\in A\sp{+}\cap CC(a)\text{\ and\ }
 a\sp{-}:=\frac12(|a|-a)\in A\sp{+}\cap CC(a).
\]
Moreover, we have $a=a\sp{+}-a\sp{-}$, $|a|=a\sp{+}+a\sp{-}$, 
and $a\sp{+}a\sp{-}=a\sp{-}a\sp{+}=0$. 

Clearly, $P\subseteq E\subseteq A$. An effect $e\in E$ is said to 
be \emph{sharp} iff the only effect $f\in E$ such that $f\leq 
e$ and $f\leq 1-e$ is $f=0$. Obviously, the unit interval $E$ is 
convex---in fact, $E$ forms a \emph{convex effect algebra} \cite
{GPBB} under the partial binary operation obtained by restriction 
to $E$ of the addition operation on $A$. By \cite[Theorem 2.6]
{FSynap}, $P$ is the set of all sharp effects, and it is also the 
set of all extreme points of the convex set $E$. 

The \emph{generalized Hermitian algebras}, introduced and studied 
in \cite{GHAlg1, GHAlg2}, are special cases of synaptic algebras; 
in fact, the synaptic algebra $A$ is a generalized Hermitian algebra 
iff it satisfies the condition that every bounded ascending sequence 
$a\sb{1}\leq a\sb{2}\leq\cdots$ of pairwise commuting elements in $A$ 
has a supremum $a$ in $A$ and $a\in CC(\{a\sb{n}:n\in\Nat\})$ \cite
[Section 6]{GHAlg1}.

If $(A\sb{i}:i\in I)$ is a nonempty family of synaptic algebras and 
$R\sb{i}$ is the enveloping algebra of $A\sb{i}$ for each $i\in I$, then 
with coordinatewise operations and partial order, the cartesian product 
{\huge$\times$}$\sb{i\in I}A\sb{i}$ is again a synaptic algebra with 
{\huge$\times$}$\sb{i\in I}R\sb{i}$ as its enveloping algebra. 

\section{Review of orthomodular lattices} \label{sc:OML} 

As we have mentioned, it turns out that the set $P$ of projections 
in the synaptic algebra $A$ forms an orthomodular lattice (OML) 
\cite{Beran, Kalm}; hence we devote this section to a brief review 
of some of the theory of OMLs that we shall require in what follows.

Let $L$ be a nonempty set partially ordered by $\leq$. If there is 
a smallest element, often denoted by $0$, and a largest element, 
often denoted by $1$, in $L$, then we say that $L$ is \emph{bounded}. 
If, for every $p,q\in L$, the \emph{meet} $p\wedge q$ (i.e., the 
greatest lower bound, or infimum) and the \emph{join} $p\vee q$ 
(i.e., the least upper bound, or supremum) of $p$ and $q$ exist in 
$L$, then $L$ is called a \emph{lattice}. If $L$ is a bounded 
lattice, then elements $p,q\in L$ are said to be \emph{complements} 
of each other iff $p\wedge q=0$ and $p\vee q=1$.

If every subset of $L$ has an infimum and a supremum, then $L$ is 
called a \emph{complete} lattice. A subset $S$ of $L$ is said to 
be \emph{sup/inf-closed} in $L$ iff whenever a nonempty subset $Q$ 
of $S$ has a supremum $s:=\bigvee Q$ (respectively, an infimum 
$t:=\bigvee Q$) in $P$, then $s\in S$, whence $s$ is the supremum of 
$Q$ as calculated in $S$ (respectively, $t\in S$, whence $t$ is the 
infimum of $Q$ as calculated in $S$).

Let $L$ be a bounded lattice. A mapping $p\mapsto p\sp{\perp}$ on $L$ 
is called an \emph{orthocomplementation} iff, for all $p,q\in L$, 
(i) $p\sp{\perp}$ is a complement of $p$ in $L$, (ii) $(p\sp{\perp})
\sp{\perp}=p$, and $p\leq q\Rightarrow q\sp{\perp}\leq p\sp{\perp}$. 
We say that $L$ is an \emph{orthomodular lattice} (OML) iff it is 
equipped with an orthocomplementation $p\mapsto p\sp{\perp}$ that  
satisfies the \emph{orthomodular law}: $p\leq q\Rightarrow q=p\vee
(q\wedge p\sp{\perp})$ for all $p,q\in L$. If $L$ is an OML, then 
elements $p,q\in L$ are \emph{orthogonal}, in symbols $p\perp q$, 
iff $p\leq q\sp{\perp}$. \emph{For the remainder of this section, 
we assume that $L$ is an OML.}

The following \emph{De Morgan\,duality} holds in $L$: If $Q\subseteq 
L$ and the supremum $\bigvee Q$ (respectively, the infimum $\bigwedge 
Q$) exists in $L$, then $(\bigvee Q)\sp{\perp}=\bigwedge\{q\sp{\perp}: 
q\in Q\}$ (respectively, $(\bigwedge Q)\sp{\perp}=\bigvee\{q\sp{\perp}: 
q\in Q\}$). 

The elements $p,q\in L$ are said to be (\emph{Mackey}) \emph{compatible} 
in $L$ iff there are pairwise orthogonal elements $p\sb{1}, q\sb{1}, 
d\in L$ such that $p=p\sb{1}\vee d$ and $q=q\sb{1}\vee d$. For instance, 
if $p\leq q$, or if $p\perp q$, then $p$ and $q$ are compatible; also, 
if $p$ and $q$ are compatible, then so are $p$ and $q\sp{\perp}$. 
As is well-known (e.g., see \cite[Proposition 4, p. 24]{Kalm} or 
\cite[Prop. 1.3.8]{PtPu}), compatibility is preserved under the 
formation of arbitrary existing suprema or infima in $L$. Computations 
in $L$ are facilitated by the following result: If $p,q,r\in L$ and one 
of the elements $p,q$, or $r$ is compatible with the other two, then 
the distributive relations $(p\vee q)\wedge r=(p\wedge r)\vee(q\wedge r)$ 
and $(p\wedge q)\vee r=(p\vee r)\wedge(q\vee r)$ hold \cite{OMLNote}.

The subset of $L$ consisting of all elements of $L$ that are compatible 
with every element of $L$ is called the \emph{center} of $L$. As is 
well-known \cite[p, 26]{Kalm}, the center of $L$ forms a boolean 
algebra, i.e., a bounded, complemented, distributive lattice 
\cite{Sikor}, and it is sup/inf-closed in $L$.

For each $p\in L$, the mapping $\phi\sb{p}\colon P\to P$ defined 
for $q\in L$ by $\phi\sb{p}q:=p\wedge(p\sp{\perp}\vee q)$ is called 
the \emph{Sasaki projection} corresponding to $p$. The Sasaki 
projection has the following properties for all $p,q,r\in L$: 
(i) $\phi\sb{p}q\perp r\Leftrightarrow q\perp\phi\sb{p}r$. (ii) 
$\phi\sb{p}\colon P\to P$ is order preserving. (iii) $\phi\sb{p}
(\phi\sb{p}q)=\phi\sb{p}q$. (iv) $p$ and $q$ are compatible iff $
\phi\sb{p}q=p\wedge q$ iff $\phi\sb{p}q\leq q$. (v) $p\perp q$ 
iff $\phi\sb{p}q=0$. (vi) $\phi\sb{p}$ preserves arbitrary 
existing suprema in $L$.

If $p\in L$, the $p$-\emph{interval}, defined and denoted by  
$L[0,p]:=\{q\in L:0\leq q\leq p\}$, is a sublattice of $L$ with 
greatest element $p$ and it forms an OML in its own right with 
$q\mapsto q\sp{\perp\sb{p}}=q\sp{\perp}\wedge p$ as the 
orthocomplementation. If $c$ belongs to the center of $L$, it 
is easy to see that $c\wedge p$ belongs to the center of 
$L[0,p]$. If, conversely, for every $p\in L$, every element $d$ 
of the center of $L[0,p]$ has the form $d=c\wedge p$ for some 
$c$ in the center of $L$, then $L$ is said to have the 
\emph{relative center property} 
\cite{Chev}.

\begin{lemma} \label{le:pintervalSasaki}
Let $p\in L$, let $q,r\in L[0,p]$, and let $\phi\sp{p}
\sb{q}\colon L[0,p]\to L[0,p]$ be the Sasaki projection 
determined by $q$ on the OML $L[0,p]$. Then{\rm: (i)} $\phi
\sp{p}\sb{q}r=\phi\sb{p}r$, i.e., $\phi\sp{p}\sb{q}$ is 
the restriction to $L[0,p]$ of the Sasaki projection $\phi
\sb{p}$ on $L$. {\rm(ii)} $\phi\sp{p}\sb{q}(r\sp{\perp\sb{p}})
=\phi\sb{p}(r\sp{\perp})$.
\end{lemma}

\begin{proof} 
(i) Since $r=r\wedge p$, $q=q\wedge p$, and $p$ is compatible with 
both $q\sp{\perp}$ and $r$, we have 
$\phi\sp{p}\sb{q}r=q\wedge(q\sp{\perp\sb{p}}\vee r)=
q\wedge((q\sp{\perp}\wedge p)\vee(r\wedge p))=q\wedge p
\wedge(q\sp{\perp}\vee r)=q\wedge(q\sp{\perp}\vee r)
=\phi\sb{q}(r)$. 

(ii) Similarly, $\phi\sp{p}\sb{q}(r\sp{\perp\sb{p}})=q
\wedge(q\sp{\perp\sb{p}}\vee r\sp{\perp\sb{p}})=
q\wedge((q\sp{\perp}\wedge p)\vee(r\sp{\perp}\wedge p))=
q\wedge p\wedge(q\sp{\perp}\vee r\sp{\perp})=q\wedge
(q\sp{\perp}\vee r\sp{\perp})=\phi\sb{q}(r\sp{\perp})$.
\end{proof}

If $p$ and $q$ share a common complement in $L$, they are said to be 
\emph{perspective} and if $p$ and $q$ are perspective in $L[0,p
\vee q]$, they are said to be \emph{strongly perspective}. Strongly 
perspective elements are perspective, but in general, not conversely. 
In fact, $L$ is modular (i.e., for all $p,q,r\in L$, $p\leq r
\Rightarrow p\wedge(q\vee r)=(p\wedge q)\vee r$) iff perspective 
elements in $L$ are always strongly perspective \cite[Theorem 2]
{SSH64}. The transitive closure of the relation of perspectivity 
is an equivalence relation on $L$ called \emph{projectivity}. If 
$L$ is modular and complete as a lattice, then by classic results 
of von Neumann \cite{vNcg} and Kaplansky \cite{Kapcg}, perspectivity 
is transitive on $L$, and therefore it coincides with projectivity.

Proof of the following lemma is a straightforward OML-calculation.

\begin{lemma}\label{le:relcompl} 
If $e,f,p\in L$ and if $e$ and $f$ are perspective in $L[0,p]$, 
then $e$ and $f$ are perspective in $L$. In fact, if $q\in 
L[0,p]$ is a common complement of $e$ and $f$ in $L[0,p]$, then 
$q\vee p\sp{\perp}$ is a common complement of $e$ and $f$ in $L$.
\end{lemma}

If $p,q\in L$, we have the \emph{parallelogram law} asserting that 
$(p\vee q)\wedge p\sp{\perp}=\phi\sb{p\sp{\perp}}(q)$ is strongly 
perspective to $q\wedge(p\wedge q)\sp{\perp}=\phi\sb{q}(p\sp
{\perp})$ (see the proof of \cite[Corollary 1]{SSH64}). Replacing 
$p$ by $p\sp{\perp}$, we obtain an alternative version of the 
parallelogram law asserting that $\phi\sb{p}(q)$ is strongly 
perspective to $\phi\sb{q}(p)$. (Another version of the 
parallelogram law is given in Theorem \ref{th:Sym&Sasaki} (ii) 
below.)

The following theorem provides an analogue for strong perspectivity 
of \cite[Lemma 43]{Loom} for a dimension equivalence relation on 
an OML.
 
\begin{theorem} \label{th:pqef}
Suppose $p,q,e,f\in P$ with $p\perp q$, $e\perp f$ and $p\vee q=
e\vee f$. Put $p\sb{1}:=p\wedge(p\wedge f)\sp{\perp}$, $p\sb{2}:=
p\wedge f$, $q\sb{1}:=q\wedge e$, $q\sb{2}:=q\wedge(q\wedge e)
\sp{\perp}$, $e\sb{1}:=e\wedge(e\wedge q)\sp{\perp}$, and 
$f\sb{2}:=f\wedge(f\wedge p)\sp{\perp}$. Then{\rm: (i)} $p\sb{1}$ 
and $e\sb{1}$ are strongly perspective. {\rm(ii)} $q\sb{2}$ and 
$f\sb{2}$ are strongly perspective. {\rm(iii)} $p\sb{1}\perp 
p\sb{2}$ with $p\sb{1}\vee p\sb{2}=p$, $q\sb{1}\perp e\sb{1}$  
with $q\sb{1}\vee e\sb{1}=e$, $p\sb{2}\perp f\sb{2}$ with $p\sb{2}
\vee f\sb{2}=f$, and $q\sb{1}\perp q\sb{2}$ with $q\sb{1}\vee q
\sb{2}=q$. {\rm(iv)} $p\sb{1}\perp q\sb{1}$ and $p\sb{1}\vee q\sb{1}$ 
is strongly perspective to $e$. {\rm(v)} $p\sb{2}\perp q\sb{2}$ and  
$p\sb{2}\vee q\sb{2}$ is strongly perspective to $f$. 
\end{theorem}

\begin{proof}
(i) Let $k:=p\vee q=e\vee f$. By the parallelogram law in the OML 
$L[0,k]$, and with the notation and results of Lemma \ref
{le:pintervalSasaki}, we find that $\phi\sp{k}\sb{p}(e)=\phi\sp{k}
\sb{p}(f\sp{\perp\sb{k}})=\phi\sb{p}(f\sp{\perp})=p\wedge(p\sp{\perp}
\vee f\sp{\perp})=p\wedge(p\wedge f)\sp{\perp}=p\sb{1}$ and $\phi\sp{k}
\sb{e}(p)=\phi\sp{k}\sb{e}(q\sp{\perp\sb{k}})=\phi\sb{e}(q\sp{\perp})
=e\wedge(e\sp{\perp}\vee q\sp{\perp})=e\wedge(e\wedge q)\sp{\perp}=
e\sb{1}$ have a common complement in $(L[0,k])[0,p\sb{1}\vee e\sb{1}]
=L[0,p\sb{1}\vee e\sb{1}]$, proving (i). 

(ii) Since $k=f\vee e=q\vee p$, (ii) follows from (i) by symmetry.

(iii) The assertions in (iii) are obvious.

(iv) We have $p\sb{1}\leq p\leq q\sp{\perp}\leq q\sp{\perp}\vee e
\sp{\perp}=q\sb{1}\sp{\ \perp}$, so $p\sb{1}\perp q\sb{1}$. By (i) 
there exists $v\sb{1}\in L[0,p\sb{1}\vee e\sb{1}]$ such that 
$p\sb{1}\vee v\sb{1}=e\sb{1}\vee v\sb{1}=p\sb{1}\vee e\sb{1}$ and 
$p\sb{1}\wedge v\sb{1}=e\sb{1}\wedge v\sb{1}=0$. We claim that 
$v\sb{1}$ is also a common complement of $p\sb{1}\vee q\sb{1}$ and 
$e$ in $L[0,(p\sb{1}\vee q\sb{1})\vee e]$. We note that $(p\sb{1}
\vee q\sb{1})\vee e=p\sb{1}\vee q\sb{1}\vee e\sb{1}\vee q\sb{1}=
p\sb{1}\vee q\sb{1}\vee e\sb{1}=p\sb{1}\vee e$. Also, $(p\sb{1}
\vee q\sb{1})\vee v\sb{1}=(p\sb{1}\vee v\sb{1})\vee q\sb{1}=
p\sb{1}\vee e\sb{1}\vee q\sb{1}=p\sb{1}\vee e$ and $e\vee v\sb{1}=
q\sb{1}\vee e\sb{1}\vee v\sb{1}=q\sb{1}\vee p\sb{1}\vee e\sb{1}=
p\sb{1}\vee e$. Moreover, $p\sb{1}\leq q\sb{1}\sp{\ \perp}$ and 
$e\sb{1}=e\wedge(e\wedge q)\sp{\perp}\leq(e\wedge q)\sp{\perp}=q
\sb{1}\sp{\ \perp}$, whence $v\sb{1}\leq p\sb{1}\vee e\sb{1}
\leq q\sb{1}\sp{\ \perp}$, and we have $v\sb{1}\perp q\sb{1}$.
Thus, since $q\sb{1}$ is orthogonal to, hence compatible with, 
both $p\sb{1}$ and $v\sb{1}$, it follows that $(p\sb{1}\vee q
\sb{1})\wedge v\sb{1}=(p\sb{1}\wedge v\sb{1})\vee(q\sb{1}\wedge 
v\sb{1})=0$. Similarly, as $q\sb{1}$ is orthogonal to both 
$e\sb{1}$ and $v\sb{1}$, we have $(e\sb{1}\vee q\sb{1})\wedge v
\sb{1}=(e\sb{1}\wedge v\sb{1})\vee(q\sb{1}\wedge v\sb{1})=0$, 
completing the proof of our claim.

(v) Since $k=q\vee p=f\vee e$, (v) follows from (iv) by symmetry.
\end{proof}

\begin{remark} \label{rm:OMLeffectalg}
We recall that the OML $L$ is organized into an effect algebra 
\cite[p. 284]{COEA} in which every element is principal \cite
[p. 286]{COEA} by defining the orthosum $p\oplus q:=p\vee q$ 
of $p$ and $q$ in $L$ iff $p\perp q$. Then the effect-algebra 
partial order coincides with the partial order on $L$ and the 
effect-algebra orthosupplementation is the orthocomplementation on $L$. 
Thus the theory of effect algebras is applicable to OMLs.
\end{remark}

As is easily seen, if the OML $L$ is regarded as an effect algebra, 
then a family of elements in $L$ is \emph{orthogonal} iff it is 
pairwise orthogonal, such an orthogonal family is orthosummable 
iff it has a supremum, and if the family is orthosummable, then 
its supremum is its orthosum \cite[p. 286]{COEA}. If every orthogonal 
family in an effect algebra is orthosummable, then the effect algebra 
is called \emph{orthocomplete} \cite{JPorthocompl}. If the OML $L$ is 
regarded as an effect algebra, then $L$ is orthocomplete iff it is 
complete as a lattice \cite{SSH70}. 
 
\section{The orthomodular lattice of projections} \label{sc:latproj} 

By \cite[Theorem 5.6]{FSynap}, under the partial order inherited
from $A$, the set $P$ of projections forms an orthomodular
lattice  with $p\mapsto p\sp{\perp} :=1-p$ as the 
orthocomplementation. As $P\subseteq A$, the OML $P$ 
acquires several special properties not enjoyed by OMLs 
in general. 

Let $p,q\in P$. By \cite[Theorem 2.4 and Lemma 5.3]{FSynap}, 
$p\leq q\Leftrightarrow p=pq\Leftrightarrow p=qp\Leftrightarrow 
p=qpq=J\sb{q}(p)$ and $p\leq q\Rightarrow q-p=q\wedge p\sp{\perp}$. 
Moreover, $pCq\Rightarrow pq=qp=p\wedge q$. Evidently, $p\perp q$ 
iff $p+q\leq 1$. Also by \cite[Lemma 5.3]{FSynap}, $p\perp q
\Leftrightarrow pq=qp=0$ and $p\perp q\Rightarrow p\vee q=p+q$. 
We refer to $p+q$ as the \emph{orthogonal sum} of $p$ and $q$ iff 
$p\perp q$. A simple argument yields the important result that 
$p$ and $q$ are compatible iff $pCq$ \cite[Theorem 3.11]{FPMonotone}.

By \cite[Theorems 2.7 and 2.10]{FSynap}, each element $a\in A$ has
a \emph{carrier projection} $a\dg\in P$ such that, for all $b\in A$,
$ab=0\Leftrightarrow a\dg b=0$; moreover, $a\dg\in CC(a)$, $a=aa\dg
=a\dg a$, $a\dg=|a|\dg$, and for all $b\in A$, $ab=0\Leftrightarrow a
\dg b\dg=0\Leftrightarrow b\dg a\dg=0\Leftrightarrow ba=0$. Furthermore, 
if $q\in P$, then $aq=a\Leftrightarrow qa=a\Leftrightarrow a\dg\leq q$. 
The carrier projection $a\dg$ is uniquely characterized by the property 
$ap=0\Leftrightarrow a\dg p=0$ for all $p\in P$, or equivalently, by 
the property that $a\dg$ is the smallest projection $q\in P$ such that
$a=aq$.

\begin{lemma}\label{le:car} If $a\sb{1},a\sb{2},...,a\sb{n}\in A\sp{+}$, 
then $(\sum\sb{i=1}\sp{n}a\sb{i})\dg=\bigvee\sb{i=1}\sp{n}(a\sb{i})\dg$. 
\end{lemma}

\begin{proof} By \cite[Lemma 3.1]{PuNote}, the lemma holds for the 
case $n=2$, and the general case then follows by mathematical 
induction. 
\end{proof}

If $a\in A$, then $\sgn(a):=(a\sp{+})\dg-(a\sp{-})\dg$ is called 
the \emph{signum} of $a$, and by \cite[Theorem 3.6]{FSynap},
$\sgn(a)\in CC(a)$, $(\sgn(a))\sp{2}=a\dg$, and $a=\sgn(a)|a|=
|a|\sgn(a)$, the latter formula being called the \emph{polar 
decomposition} of $a$. 

Each element $a\in A$ has a \emph{spectral resolution} 
\cite[Section 8]{FSynap}, \cite{SOUS} that both determines and is determined 
by $a$, namely the right continuous ascending family $(p\sb{a,
\lambda}:\lambda\in\reals)$ of projections in $CC(a)$ 
given by 
\[
p\sb{a,\lambda}:=1-((a-\lambda)\sp{+})\dg=(((a-\lambda)\sp{+})\dg)
\sp{\perp}\text{\ for all\ }\lambda\in\reals.
\] 
By \cite[Theorems 8.4 and 8.5]{FSynap}, $L :=\sup\{\lambda\in\reals:p\sb{a,
\lambda}=0\}\in\reals$, $U:=\inf\{\lambda\in\reals:p\sb{a,\lambda}=1\}
\in\reals$, and $a=\int\sb{L-0}\sp{U}\lambda\,dp\sb{a,\lambda}$, where 
the Riemann-Stieltjes type integral converges in norm.

By \cite[Theorem 8.3]{FPSynap}, any one of the following conditions is 
sufficient to guarantee modularity of the projection lattice $P$: (i) 
If $p,q\in P$, there exists $0<\epsilon\in\reals$ such that $\epsilon
(pqp)\dg\leq pqp$. (ii) If $p,q\in P$, then $pqp$ is an algebraic 
element of $A$; (iii) $A$ is finite dimensional over $\reals$; (iv) 
$P$ satisfies the ascending chain condition. 

Let $p\in P$. Then, according to \cite[Theorem 4.9]{FSynap},
\[
pAp :=\{pap:a\in A\}=\{a\in A:pa=ap=a\}=\{a\in A:a\dg\leq p\}
=J\sb{p}(A)
\]
is norm closed in $A$, and with the partial order inherited from
$A$, it is a synaptic algebra (degenerate if $p=0$) with $p$ as its  
order unit, $pRp$ as its enveloping algebra, and the order unit 
norm on $pAp$ is the restriction to $pAp$ of the order unit norm on 
$A$. Moreover, if $a,b\in pAp$, then $a\circ b,\,a\dg,\,|a|,\,a\sp{+},
\, a\sp{-}\in pAp$, and if $a\in A\sp{+}$, then $a\sp{1/2}\in pAp$. 
Consequently, if $a\in pAp$, then $pp\sb{a,\lambda}=p\sb{a,\lambda}p=
p\wedge p\sb{a,\lambda}$ for all $\lambda\in\reals$, and the spectral 
resolution of $a$ as calculated in $pap$ is $(pp\sb{a,\lambda}:\lambda
\in\reals)$. Clearly, the OML of projections in the synaptic 
algebra $pAp$ is the $p$-interval $P[0,p]$ in $P$, and the 
orthocomplementation on $P[0,p]$ is given by \[
q\mapsto q\sp{\perp\sb{p}}=p-q=J\sb{p}(q\sp{\perp})=pq\sp{\perp}
=q\sp{\perp}p=p\wedge q\sp{\perp}\text{\ for\ }q\in P[0,p].
\]  

Suppose that $B\subseteq A$. Then $C(B)=\bigcap\sb{b\in B}C(b)$
is norm closed in $A$, and with the partial order inherited
from $A$, $C(B)$ is a synaptic algebra with order unit $1$ and
enveloping algebra $R$. Moreover, if $a,c\in C(B)$, then $a
\circ c,\,a\dg,\,|a|,\,a\sp{+},\,a\sp{-}\in C(B)$ and $0\leq a
\Rightarrow a\sp{1/2}\in C(B)$. Consequently, if $a\in C(B)$,
then the spectral resolution of $a$ is the same whether
calculated in $A$ or in $C(B)$. Also it is clear that the OML 
of projections in the synaptic algebra $C(B)$ is just $P\cap C(B)$, 
and we have the following result.

\begin{lemma} \label{lm:SupInfQ}
If $B\subseteq A$, then $P\cap C(B)$ is sup/inf-closed in $P$. 
\end{lemma}

\begin{proof}
As $C(B)=\bigcap\sb{b\in B} C(b)$, it will be sufficient to prove
the lemma for the special case $B=\{b\}$. Thus, assume that $Q
\subseteq P\cap C(b)$ and that $h=\bigvee Q$ exists in $P$. For
the projections in the spectral resolution of $b$, we have $C(b)
\subseteq C(p\sb{b,\lambda})$, whence $Q\subseteq P\cap C(p\sb{b,
\lambda})$ for every $\lambda\in\reals$. We recall that compatibility 
is preserved under the formation of arbitrary existing suprema and 
infima.  Thus, $h\in C(p\sb{b,\lambda})$ for all $\lambda
\in\reals$, and it follows from \cite[Theorem 8.10]{FSynap} that 
$h\in C(b)$. A similar argument applies to the infimum $k$, if it 
exists in $P$.
\end{proof}

We shall make extensive use of the next theorem, often without 
explicit attribution.

\begin{theorem} \label{th:centP}
The center of $P$ is $P\cap C(P)=P\cap C(A)$.
\end{theorem} 

\begin{proof}
As two projections in $P$ are compatible iff they commute, the 
center of $P$ is $P\cap C(P)$. Clearly, $P\cap C(A)\subseteq P
\cap C(P)$. Conversely, by \cite[Theorem 8.10]{FSynap}, $P
\cap C(P)\subseteq P\cap C(A)$, so $P\cap C(A)=P\cap C(P)$. 
\end{proof}

In view of Theorem \ref{th:centP}, if we say that $c$ is a 
\emph{central projection} in $A$, we mean that $c\in P\cap C(A)$, 
or what is the same thing, that $c$ belongs to the center $P\cap 
C(P)$ of the OML $P$. As is easily seen, if $P$ is regarded as an 
effect algebra, then the center $P\cap C(P)$ of $P$ coincides with 
the effect-algebra center of $P$ \cite[p. 287]{COEA}.

\begin{remarks} \label{rm:DirectSum}
Suppose that $c$ is a central projection in $A$. Then $c
\sp{\perp}$ is also a central projection and $A$ is the (internal)
\emph{direct sum} of the synaptic algebras $cAc=cA=Ac$ and $c
\sp{\perp}Ac\sp{\perp}=c\sp{\perp}A=Ac\sp{\perp}$ in the sense
that: (1) As a vector space, $A$ is the (internal) direct sum of
the vector subspaces $cA$ and $c\sp{\perp}A$. (2) If $a=x+y$ with
$x\in cA$ and $y\in c\sp{\perp}A$, then $0\leq a\Leftrightarrow
0\leq x\text{\ and\ }0\leq y$. (3) If $a\sb{i}=x\sb{i}+y\sb{i}$ with
$x\sb{i}\in cA$ and $y\sb{i}\in c\sp{\perp}A$ for $i=1,2$, then
$a\sb{1}a\sb{2}=x\sb{1}x\sb{2}+y\sb{1}y\sb{2}$ (the products
being calculated in $R$). With coordinatewise operations and
partial order, the cartesian product $cA\times c\sp{\perp}A$ is
a synaptic algebra with order unit $(c,c\sp{\perp})$ and enveloping
algebra $cRc\times c\sp{\perp}Rc\sp{\perp}$, and $A$ is order, 
linear, and Jordan isomorphic to $cA\times c\sp{\perp}A$ under 
the mapping $a\mapsto (ca,c\sp{\perp}a)$.
\end{remarks}

Naturally, $A$ is called a \emph{commutative} synaptic algebra 
iff $aCb$ for all $a,b\in A$, i.e., iff $A=C(A)$. Thus, a commutative 
synaptic algebra is a commutative associative archimedean partially 
ordered linear algebra with a unity element. In \cite[Section 4]{GHAlg2}, 
the following result is stated without proof. 

\begin{theorem} \label{th:Commutative}
The synaptic algebra $A$ is commutative iff $P$ is a boolean algebra i.e., 
iff the lattice $P$ is distributive.
\end{theorem}

\begin{proof}
If $A$ is commutative, then $P\cap C(A)=P\cap A=P$, i.e., $P$ is 
its own center, whence $P$ is a boolean algebra. Conversely, if $P$ 
is a boolean algebra, then any two projections $p,q\in P$ are 
compatible, and therefore $p,q\in P\Rightarrow pCq$. Let $a,b\in A$. 
Then, for all $\lambda,\mu\in\reals$, $p\sb{a,\lambda}Cp\sb{b,\mu}$, 
whence $aCp\sb{b,\mu}$ by \cite[Theorem 8.10]{FSynap}. Therefore, 
$aCb$ by a second application of \cite[Theorem 8.10]{FSynap}, and it 
follows that $A$ is commutative.
\end{proof}

It can be shown that every boolean algebra can be realized as the 
lattice of projections in a commutative synaptic algebra. The center 
$C(A)$ is a commutative synaptic algebra, and if $B$ is a subset of $A$
consisting of pairwise commuting elements, then $CC(B)$ is a
commutative synaptic algebra. In particular, $CC(a)$ is a 
commutative synaptic algebra for any choice of $a\in A$.

\section{Symmetries and perspectivities} \label{sc:sym&persp} 

Although there is some overlap between this section and 
\cite[Section 3]{PuNote}, the material here is arranged a little 
differently, so for the reader's convenience, we give proofs 
of most of our results.

\begin{definition} An element $s\in A$ is called a \emph{symmetry} 
iff $s\sp{2}=1$.  An element $t\in A$ is called a \emph{partial 
symmetry} iff $t\sp{2}\in P$.
\end{definition}

Proofs of the following statements are straightforward. (i) If $t
\in A$ is a partial symmetry with $p:=t\sp{2}$, then $t$ is a symmetry 
in the synaptic algebra $pAp$. (ii) If $a\in A$, then the element $t:=
\sgn(a)$ in the polar decomposition $a=t|a|$ is a partial symmetry 
with $t\sp{2}=a\dg$. (iii) If $s$ is a symmetry, then $-s$ is a 
symmetry as well. (iv) There is a bijective correspondence $p
\leftrightarrow s$ between symmetries in $A$ and projections in $P$ 
given by $s=2p-1$ and $p=\frac{1}{2}(1+s)$. (v) If $s$ is a symmetry, 
then $\|s\|\sp{2}=\|s\sp{2}\|=\|1\|=1$, so $\|s\|=1$. (vi) If $s$ is a 
symmetry, then as $0\leq\frac{1}{2}(1\pm s)\in P$, it follows that 
$-1\leq s\leq 1$.

If $p\in P$ and $s:=2p-1=p-(1-p)$ is the corresponding symmetry, 
then $s$ is the difference of the orthogonal projections $p$ and $1-p$. More 
generally, by the following theorem, a difference of two orthogonal projections 
$p$ and $q$ is a partial symmetry and \emph{vice versa}.

\begin{theorem} \label{th:ParsymDifProj}
If $p$ and $q$ are orthogonal projections, then $t:=p-q$ is a partial symmetry with 
$t\sp{2}=p+q=p\vee q\in P$. Conversely, if $t$ is a partial symmetry, then both 
$p:=t\sp{+}$ and $q:=t\sp{-}$ are projections, $pq=0$, $t=p-q$, $t\sp{2}=p+q=
|t|\in P$, and $s:=t+(1-t\sp{2})=t+(1-|t|)=1-2q$ is a symmetry.
\end{theorem}

\begin{proof}
The first statement of the lemma is obvious. So assume that $u:=t\sp{2}\in P$, 
$p:=t\sp{+}$ and $q:=t\sp{-}$. Then $t=p-q$ and $pq=0$. Also, as $0\leq u$ 
and $u\sp{2}=u$, we have $p+q=|t|=(t\sp{2})\sp{\frac{1}{2}}=u\sp{\frac{1}{2}}=u$.
Moreover, as $0\leq t\sp{+}=p$, $0\leq t\sp{-}=q$ and $u\in P$, we have $0\leq p
\leq p+q=u\leq 1$, so $p\in E$, and it follows from \cite[Theorem 2.4]{FSynap} 
that $p=pu=p(p+q)=p\sp{2}+pq=p\sp{2}$, so $p\in P$ and likewise, $q\in P$.
Then $s:=t+(1-t\sp{2})=p-q+(1-(p+q))=1-2q$ is the symmetry corresponding to 
the projection $1-q$.
\end{proof}

\noindent If $t$ is a partial symmetry, then we refer to the 
symmetry $s:=t+(1-t\sp{2})$ in Theorem \ref{th:ParsymDifProj} 
as the \emph{canonical extension of $t$ to a symmetry $s$}.

If $s\in A$ is a symmetry, then the quadratic mapping 
$J\sb{s}\colon A\to A$ is called the \emph{symmetry 
transformation} corresponding to $s$.

\begin{theorem}\label{th:symmetry&quadratic}
Let $s$ be a symmetry, $a,b\in A$, and $e,f\in P$. Then{\rm:} 
{\rm(i)} The symmetry transformation $J\sb{s}(a):=sas$ is an 
order, linear, and Jordan automorphism of $A$ and $(J\sb{s})
\sp{-1}=J\sb{s}$. {\rm(ii)} $J\sb{s}$ restricted to $P$ is an 
OML-automorphism of $P$. {\rm(iii)} If $ab=ba$, then $J\sb{s}
(ab)=J\sb{s}(a)J\sb{s}(b)=J\sb{s}(b)J\sb{s}(a)\in A$. {\rm(iv)} 
If $J\sb{s}(a)J\sb{s}(b)=J\sb{s}(b)J\sb{s}(a)$, then $ab=ba\in A$. 
{\rm(v)} $(J\sb{s}(a))\dg=J\sb{s}(a\dg)$. 
\end{theorem}

\begin{proof}(i) As $J\sb{s}$ is a quadratic mapping, it is both linear 
and order preserving. Also, for $a\in A$, $J\sb{s}(a\sp{2})=sa\sp{2}s=sassas
=(J\sb{s}(a))\sp{2}$, and it follows that $J\sb{s}$ is a Jordan homomorphism of 
$A$. Since $J\sb{s}(J\sb{s}(a))=ssass=a$, it follows that $J\sb{s}$ is its own 
inverse on $A$; hence it is a linear, order, and Jordan automorphism. 

(ii) If $e\in P$, then $(J\sb{s}(e))\sp{2}=J\sb{s}(e\sp{2})=J\sb{s}(e)$, so 
$J\sb{s}$ maps $P$ into (and clearly onto) $P$. Thus the restriction of $J\sb{s}$ 
to $P$ is an order automorphism of $P$, and as $J\sb{s}(1-e)=1-J\sb{s}(e)$, 
it is an OML automorphism.

(iii) If $ab=ba$, then $ab=ba=a\circ b=b\circ a$, so part (iii) follows 
from the fact that $J\sb{s}$ is a Jordan automorphism.

(iv) Assume the hypothesis of (iv). Then by (iii) with $a$ replaced by 
$J\sb{s}(a)$ and $b$ replaced by $J\sb{s}(b)$, we have $ab=J\sb{s}(J\sb{s}
(a))J\sb{s}(J\sb{s}(b))=J\sb{s}(J\sb{s}(a)J\sb{s}(b))=J\sb{s}(J\sb{s}(b)J
\sb{s}(a))=J\sb{s}(J\sb{s}(b))J\sb{s}(J\sb{s}(a))=ba$.

(v) By (ii), $J\sb{s}(a\dg)\in P$, and since $aa\dg=a\dg a= a$, (iii) implies 
that $J\sb{s}(a)J\sb{s}(a\dg)=J\sb{s}(aa\dg)=J\sb{s}(a)$. Suppose that $f\in P$ 
with $J\sb{s}(a)f=J\sb{s}(a)$. It will be sufficient to prove that $J\sb{s}(a\dg)
\leq f$. We have $fJ\sb{s}(a)=J\sb{s}(a)=J\sb{s}(a)f$. Put $e:=J\sb{s}(f)$. Then 
$e\in P$ and $J\sb{s}(e)=f$, so $J\sb{s}(e)J\sb{s}(a)=J\sb{s}(a)J\sb{s}(e)$; hence, 
by (iv), $ea=ae$, and therefore by (iii), $J\sb{s}(a)=J\sb{s}(a)f=J\sb{s}(a)J\sb{s}
(e)=J\sb{s}(ae)$, whence, $a=ae$ and therefore $a\dg\leq e$, whereupon $J\sb{s}(a\dg)
\leq J\sb{s}(e)=f$.
\end{proof}

\begin{definition}\label{de:exsym} A symmetry $s\in A$ is said to \emph{exchange 
the projections} $e,f\in P$ iff $ses=f$, i.e., iff $J\sb{s}(e)=f$. A partial 
symmetry $t$ is said to \emph{exchange the projections} $e,f\in P$ iff 
both $tet=f$ and $tft=e$ hold.
\end{definition}

Let $s$ be a symmetry and let $e,f\in P$. Clearly, $s$ exchanges $e$ and 
$f$ iff $J\sb{s}(e)=f$ iff $J\sb{s}(f)=e$ iff $s$ exchanges $f$ and $e$. 

\begin{theorem}\label{th:parsym} 
Let $t\in A$ be a partial symmetry that exchanges the projections 
$e,f\in P$ and let $s:=t+(1-t\sp{2})$ be the canonical extension of $t$ 
to a symmetry. Then $s$ exchanges $e$ and $f$.
\end{theorem}

\begin{proof} Assume the hypotheses and let $u:=t\sp{2}$. Then $u\in P$ and 
we have $e=tft=t\sp{2}et\sp{2}=ueu$, so $e=ue=eu$, and therefore $(1-u)e=
e(1-u)=0$. Consequently, $ses=(t+(1-u))e(t+(1-u))=tet=f$.
\end{proof}

The following lemma provides a weak version of finite additivity for 
the relation of exchangeability by a symmetry.

\begin{lemma}\label{le:finadit} 
Let $e,e\sb{1},e\sb{2},f,f\sb{1},f\sb{2}\in P$ with $e\sb{1}\perp f
\sb{2}$, $e\sb{2}\perp f\sb{1}$, $e\sb{1}\perp e\sb{2}$, $f\sb{1}
\perp f\sb{2}$, $e=e\sb{1}+e\sb{2}$ and $f=f\sb{1}+f\sb{2}$, and 
suppose that $e\sb{i}$ and $f\sb{i}$ are exchanged by a symmetry 
$s\sb{i}\in A$ for $i=1,2$. Then there is a symmetry $s\in A$ 
exchanging $e$ and $f$.
\end{lemma}

\begin{proof} 
Let $p\sb{i}=e\sb{i}\vee f\sb{i}$ for $i=1,2$. As $e\sb{1}
\leq e\sb{2}\sp{\,\perp},f\sb{2}\sp{\,\perp}$, we have 
$e\sb{1}\leq p\sb{2}\sp{\,\perp}$. Also, $f\sb{1}\leq
e\sb{2}\sp{\,\perp},f\sb{2}\sp{\,\perp}$, so $f\sb{1}\leq 
p\sb{2}\sp{\,\perp}$, and it follows that $p\sb{1}=e\sb{1}
\vee f\sb{1}\leq p\sb{2}\sp{\,\perp}$, whence $p\sb{1}p
\sb{2}=0$.  Set $u=s\sb{1}p\sb{1}$ and $v=s\sb{2}p\sb{2}$. 
From $s\sb{1}p\sb{1}s\sb{1}=s\sb{1}(e\sb{1}\vee f\sb{1})s
\sb{1}=s\sb{1}e\sb{1}s\sb{1}\vee s\sb{1}f\sb{1}s\sb{1}=e
\sb{1}\vee f\sb{1}=p\sb{1}$, it follows that $s\sb{1}$ commutes 
with $p\sb{1}$. Likewise $s\sb{2}$ commutes with $p\sb{2}$, whence 
both $u$ and $v$ belong to $A$ and are partial symmetries with 
$u\sp{2}=p\sb{1}$, $v\sp{2}=p\sb{2}$, $ue\sb{1}u=p\sb{1}f\sb{1}p
\sb{1}=f\sb{1}$, $ve\sb{2}v=p\sb{2}f\sb{2}p\sb{2}=f\sb{2}$, and 
$uv=0$. Straightforward calculation using the data above shows 
that $s:=u+v+(1-p\sb{1}-p\sb{2})$ is a symmetry and that $ses=f$.
\end{proof}

\begin{lemma} \label{lm:commuting}
If $e,f\in P$ and $a:=e-f$, then $e,f\in C(|a|)$. 
\end{lemma}

\begin{proof}
We have $a\sp{2}=(e-f)\sp{2}=e-ef-fe+f$, from which it follows that $ea\sp{2}
=a\sp{2}e=e-efe$ and $fa\sp{2}=a\sp{2}f=f-fef$. Thus, $e,f\in C(a\sp{2})$, 
and as $|a|\in CC(a\sp{2})$ it follows that $e,f\in C(|a|)$.
\end{proof}

\begin{theorem} \label{th:Symefe&fef}
If $e,f\in P$, there exists a symmetry $s\in A$ such that $sefes=fef$, 
i.e., $J\sb{s}(efe)=J\sb{s}(J\sb{e}(f))=J\sb{f}(e)=fef$.
\end{theorem}

\begin{proof}
Let $e,f\in P$. By Lemma \ref{lm:commuting} with $f$ replaced by $1-f$ and 
$a:=e+f-1$, we have $e,f\in C(|a|)$. Put $t:=\sgn(a)$, so that $t\sp{2}
=a\dg\in P$. Thus, $t$ is a partial symmetry with $|a|=at=ta$, and 
we have      
\[
|a|f=taf=(te-t(1-f))f=tef\text{\ and\ }f|a|=fat=f(et-(1-f)t)=fet.
\]
Therefore, since $ea=e+ef-e=ef$ and $|a|$ commutes with $e$ 
and $f$, 
\[
t(efe)t=(tef)et=|a|fet=|a|f|a|=f|a||a|=fet|a|=fea=fef.
\]
Now let $s:=t+(1-t\sp{2})$ be the canonical extension of $t$ to a 
symmetry (Theorem \ref{th:ParsymDifProj}). Since $tef=|a|f$, and 
$af=ef+f-f=ef$ it follows that
\[
t\sp{2}efe=t(tef)e=t|a|fe=afe=efe,
\]
so $(1-t\sp{2})efe=0$, and therefore $sefes=tefet=fef$. 
\end{proof}

\begin{theorem} \label{th:Sym&Sasaki}
Let $e,f\in P$. Then{\rm:}
\begin{enumerate}
\item  $\phi\sb{e}f$ and $\phi\sb{f}e$ are exchanged by a symmetry in $A$. 
\item {\rm (}Symmetry Parallelogram Law{\rm)} $e-(e\wedge f))$ and 
 $(e\vee f)-f$ are exchanged by a symmetry in $A$. 
\item If $e$ and $f$ are complements in $P$, then $e$ and $f\sp{\perp}$ 
are exchanged by a symmetry in $A$.
\item If $e\not\perp f$, there are nonzero subelements $0\not=e\sb{1}
\leq e$ and $0\not=f\sb{1}\leq f$ that are exchanged by a symmetry.
\end{enumerate}
\end{theorem}

\begin{proof}
(i) According to \cite[Definition 4.8 and Theorem 5.6]{FSynap}, 
$(efe)\dg=\phi\sb{e}(f)$ and $(fef)\dg=\phi\sb{f}(e)$.  Also, by 
Theorem \ref{th:Symefe&fef}, there exists a symmetry $s\in A$ such 
that $J\sb{s}(efe)=fef$, whence by Theorem \ref{th:symmetry&quadratic} 
we have $J\sb{s}(\phi\sb{e}f)=J\sb{s}((efe)\dg)=(J\sb{s}(efe))\dg=
(fef)\dg=\phi\sb{f}e$.

(ii) We have $e\wedge(e\wedge f)\sp{\perp}=e\wedge(e\sp{\perp}
\vee f\sp{\perp})=\phi\sb{e}(f\sp{\perp})$ and $(e\vee f)\wedge f
\sp{\perp}=\phi\sb{f\sp{\perp}}(e)$, so (ii) follows from (i).

(iii) If $e\wedge f=0$ and $e\vee f=1$, then the symmetry $s$ in 
(ii) satisfies $ses=J\sb{s}(e)=f\sp{\perp}$.

(iv) If $e\not\perp f$, then $0\not=e\sb{1}:=\phi\sb{e}f\leq e$ 
and $0\not=f\sb{1}:=\phi\sb{f}(e)\leq f$, and by (i), $e\sb{1}$ and 
$f\sb{1}$ are exchanged by a symmetry.
\end{proof}

\begin{lemma} \label{le:compsym}
If the projections $e$ and $f$ are complements in $P$ and 
$s\in A$ is a symmetry exchanging $e$ and $f$, then $e$ and 
$f$ are perspective. In fact, the projection $p:=\frac{1}{2}(1+s)
\in P$ corresponding to the symmetry $s$ is a common complement 
of both $e$ and $f$. 
\end{lemma}

\begin{proof}
Let $q:=e\wedge p$. Then $q\leq e$ and $q\leq p$, so $eq=pq=q$. 
Thus, $sq=(2p-1)q=q$, so $fq=sesq=seq=sq=q$, whence $q\leq f$. 
But $q\leq e$, so $e\wedge p=q\leq e\wedge f=0$. Now let $r:=
e\sp{\perp}\wedge p\sp{\perp}$. Then $er=pr=0$, so $sr=(2p-1)r
=-r$, whence $fr=sesr=-ser=0$, and we have $r\leq f\sp{\perp}$. 
But $r\leq e\sp{\perp}$, so $r\leq e\sp{\perp}\wedge f\sp{\perp}
=0$. Therefore, $p$ is a complement of $e$, and by a similar 
argument, $p$ is also a complement of $f$.  
\end{proof}

In the following theorem we improve the result in Lemma 
\ref{le:compsym} by dropping the hypothesis that $e$ and 
$f$ are complements, and by concluding that $e$ and $f$ 
are not only perspective, but strongly perspective.

\begin{theorem} \label{th:exstrongpersp} 
Let $e,f\in P$ be exchanged by a symmetry $s\in A$. 
Then $e$ and $f$ are strongly perspective in $P$. In fact, 
if $p:=e\vee f$, $r:=p-(e\wedge f)$, $t:=rsr$, and $q:=
\frac{1}{2}(r+t)$, then $t$ is a symmetry in $rAr$, $q$ is 
a projection in $P[0,r]$, and $k:=q\vee(r\sp{\perp}\wedge p)$ 
is a common complement of $e$ and $f$ in $P[0,p]$.
\end{theorem}

\begin{proof}
Assume the hypotheses. By Theorem \ref{th:symmetry&quadratic} 
(ii), $sps=s(e\vee f)s=J\sb{s}(e\vee f)=J\sb{s}(e)\vee J\sb{s}(f)
=ses\vee sfs=f\vee e=p$, whence $sp=ps$. Likewise, $s(e\wedge f)s
=e\wedge f$ and it follows that $srs=r$, whence $t=rsr=sr=rs$. 
Therefore, $t\in rAr$ with $t\sp{2}=r$, i.e., $t$ is a symmetry 
in the synaptic algebra $rAr$. Clearly, both $e$ and $f$ commute 
with $r$, whence $t(e\wedge r)t=rs(er)sr=rsesr=rfr=fr=rf=f\wedge 
r$, and therefore $t$ exchanges the projections $e\wedge r$ and 
$f\wedge r$ in $P[0,r]$.

We have $(e\wedge r)\wedge(f\wedge r)=(e\wedge f)\wedge p\wedge 
(e\wedge f)\sp{\perp}=0$, and as $r\leq p$, $(e\wedge r)
\vee(f\wedge r)=(e\vee f)\wedge r=p\wedge r=r$, so $e\wedge r$ 
and $f\wedge r$ are complements in $P[0,r]$. Thus, working 
in the synaptic algebra $rAr$ with unit $r$, and applying Lemma 
\ref{le:compsym}, we find that $e\wedge r$ and $f\wedge r$ are 
perspective in $P[0,r]$ with $q=\frac{1}{2}(r+t)$ as a common 
complement. Therefore, 
\[
(e\wedge r)\vee k=(e\wedge r)\vee q\vee(r\sp{\perp}\wedge p))
=r\vee(p-r)=p.
\]
Also, as $q\leq r\leq p$ and $e\wedge r\leq r\leq p$, it 
follows that both $q$ and $e\wedge r$ commute with 
$r\sp{\perp}\wedge p$, whence 
\[
(e\wedge r)\wedge k=(e\wedge r)\wedge(q\vee(r\sp{\perp}\wedge p))
 =((e\wedge r)\wedge q)\vee(e\wedge r\wedge r\sp{\perp}\wedge p)
 =0.
\]
Likewise, $(f\wedge r)\vee k=p$ and $(f\wedge r)\wedge k=0$. 
\end{proof}

\begin{theorem}\label{th:perspex} 
Let $e,f\in P$. Then{\rm:}
\begin{enumerate}
\item If $e$ and $f$ are perspective, then there are symmetries 
 $s\sb{1},s\sb{2}\in A$ with $s\sb{2}s\sb{1}es\sb{1}s\sb{2}=
 J\sb{s\sb{2}}(J\sb{s\sb{1}}(e))=f$.
\item Suppose $e$ and $f$ are orthogonal and there are 
 symmetries $s\sb{1},s\sb{2}\in A$ with $s\sb{2}s\sb{1}es
 \sb{1}s\sb{2}=f$. Then there is a symmetry $s$ exchanging 
 $e$ and $f$.
\item If $e$ and $f$ are both perspective and orthogonal, then 
 there is a symmetry $s$ exchanging $e$ and $f$.
\end{enumerate}
\end{theorem}
\begin{proof} (i) Let $p$ be a common complement of $e$ and $f$. 
By Theorem \ref{th:Sym&Sasaki} (iii), there exist symmetries $s\sb{1},
s\sb{2}\in A$ with $s\sb{1}es\sb{1}=p\sp{\perp}=s\sb{2}fs\sb{2}$, 
and it follows that $s\sb{2}s\sb{1}es\sb{1}s\sb{2}=f$.

(ii) Assume the hypotheses of (ii). Then $e=s\sb{1}s\sb{2}
fs\sb{2}s\sb{1}$. Let 
\[
x:=s\sb{2}s\sb{1}e\text{\ and\ }y:=es\sb{1}s\sb{2}.
\]
Here $x$ and $y$ belong to the enveloping algebra $R$, but not 
necessarily to $A$; however, $x+y\in A$ by Lemma \ref{le:abc}. 
As $ef=fe=0$, it follows that $xf=fy=0$. We have $xy=f$, $yx=e$, 
$xe=x$, and $fx=s\sb{2}s\sb{1}es\sb{1}s\sb{2}s\sb{2}s\sb{1}e=x$, 
so $x^2=xefx=0$. Also, $ey=y$, and $yf=es\sb{1}s\sb{2}s\sb{2}s
\sb{1}es\sb{1}s\sb{2}=y$, so $y\sp{2}=yfey=0$. Furthermore, 
$ye=es\sb{1}s\sb{2}e=s\sb{1}s\sb{2}fs\sb{2}s\sb{1}s\sb{1}s
\sb{2}e=s\sb{1}s\sb{2}fe=0$ and $ex=es\sb{2}s\sb{1}e=es
\sb{2}s\sb{1}s\sb{1}s\sb{2}fs\sb{2}s\sb{1}=efs\sb{2}s\sb{1}
=0$.  

Now put $s:=(x+y)+1-e-f$. Using the data above, a straightforward 
computation shows that $s$ is a symmetry in $A$ and  $s$ 
exchanges $e$ and $f$.

(iii) Part (iii) follows from (i) and (ii).
\end{proof}

Theorem \ref{th:weakadditivity} below, is a version of 
additivity for projections exchanged by a symmetry, 
but it requires completeness of $P$ and the rather strong 
hypothesis that the suprema of the two orthogonal families 
involved are themselves orthogonal. The next two lemmas will 
aid in its proof.

\begin{lemma} \label{lm:prep1}
Let $e,f\in P$ with $e\perp f$, suppose that $e$ and $f$ are 
exchanged by a symmetry $s\in A$, put $x:=se$, $y:=es$, and 
$p:=\frac{1}{2}(x+y+e+f)$. Then{\rm: (i)} $xy=f$ and $yx=e$. 
{\rm(ii)} $x=xe=fx$, $y=ey=yf$, and $x\sp{2}=y\sp{2}=ex=xf=ye
=fy=0$. {\rm(iii)} $p\in P$. {\rm(iv)} $2epe=e$ and $2pep=p$. 
{\rm(v)} $2fpf=f$ and $2pfp=p$.
\end{lemma}

\begin{proof}
(i) We have $ses=f$, $sfs=e$, so $xy=f$. Also $yx=es\sp{2}e=
e\sp{2}=e$. 

(ii) Clearly, $xe=x$, $ey=y$, and since $ef=fe=0$, $xf=0$, 
and $fy=0$. Moreover, $x\sp{2}=sese=fe=0$, $y\sp{2}=eses=
fs=fes=0$, $ex=ese=sfsse=sfe=0$, $yf=esf=esses=es=y$, and 
similarly, $y\sp{2}=0$, $ye=0$, and $fx=x$. 

(iii) We have $x+y\in A$, whence $p\in A$. A straightforward 
computation using the data in (ii) shows that $p\sp{2}=p$.

(iv) Since $ex=ye=0$ and $ef=0$, it follows that $2epe=exe+
eye+e+efe=e$. Similarly, $2pe=xe+ye+e=x+e$, so $2pep=xp+ep=
\frac{1}{2}(x\sp{2}+xy+xe+xf+ex+ey+e)=\frac{1}{2}(f+x+y+e)=p$. 

(v) As in the proof of (iii), the proof of (iv) is a 
straightforward computation using the data in (ii).
\end{proof}

\begin{lemma} \label{lm:prep2}
Suppose that $P$ is a complete OML, let $(q\sb{i})\sb{i\in I}$ 
be an orthogonal family in $P$, put $q:=\bigvee\sb{i\in I}q
\sb{i}$, let $i\in I$, $r\in P$, and suppose that $q\sb{j}r=0$ 
for all $j\in I$ with $j\not=i$. Then $qr=q\sb{i}r$ and $rq=rq\sb{i}$.
\end{lemma}

\begin{proof}
Define $q\sb{i}':=\bigvee\sb{j\in I,\,j\not=i}q\sb{j}$. 
As $q\sb{j}\perp q\sb{i}$ for all $j\in I$ with $j\not=i$, 
it follows that $q\sb{i}\perp q\sb{i}'$ with $q=q\sb{i}
\vee q\sb{i}'=q\sb{i}+q\sb{i}'$. Also, since $q\sb{j}\perp 
r$ for $j\in I$ with $j\not=i$, it follows that $q\sb{i}'
\perp r$, whence $q\sb{i}'r=rq\sb{i}'=0$, and therefore 
$qr=(q\sb{i}+q\sb{i}')r=q\sb{i}r$ and $rq=r(q\sb{i}+
q\sb{i}')=rq\sb{i}$. 
\end{proof}

\begin{theorem}\label{th:weakadditivity} 
Suppose that the OML $P$ is complete, let $(e\sb{i})\sb
{i\in I}$ and $(f\sb{i})\sb{i\in I}$ be orthogonal families 
in $P$ with $e=\bigvee\sb{i\in I}e\sb{i}$ and $f=\bigvee
\sb{i\in I}f\sb{i}$. Then, if $e$ and $f$ are orthogonal and 
if, for each $i\in I$, $e\sb{i}$ and $f\sb{i}$ are exchanged 
by a symmetry $s\sb{i}\in A$, then there is a symmetry $s\in A$ 
exchanging $e$ and $f$.
\end{theorem}

\begin{proof} Our proof is suggested by the proof of 
\cite[Lemma 3.1]{Kappro}. We begin by noting that $e\perp f$ 
implies $e\sb{i}\perp f\sb{j}$ for all $i,j\in I$. Also, for 
$i\in I$, we have $s\sb{i}e\sb{i}s\sb{i}=f\sb{i}$ and $s\sb{i}
f\sb{i}s\sb{i}=e\sb{i}$. Let $x\sb{i}:=s\sb{i}e\sb{i}$, $y
\sb{i}:=e\sb{i}s\sb{i}$, and $p\sb{i}:=\frac{1}{2}(x\sb{i}+y
\sb{i}+e\sb{i}+f\sb{i})$. By parts (i) and (iii) of Lemma 
\ref{lm:prep1}, $y\sb{i}x\sb{i}=e\sb{i}$, $x\sb{i}y\sb{i}=f
\sb{i}$, and $p\sb{i}\in P$. Put $p:=\bigvee\sb{i\in I}
p\sb{i}$ and $s:=2p-1$. We are going to show that $s$ is the 
required symmetry.

We claim that $(p\sb{i})\sb{i\in I}$ is an orthogonal family. 
Indeed, suppose $i,j\in I$ with $i\not=j$. Then $4p\sb{i}p
\sb{j}=(x\sb{i}+y\sb{i}+e\sb{i}+f\sb{i})(x\sb{j}+y\sb{j}+e
\sb{j}+f\sb{j})$ and it will be sufficient to show that the 
sixteen terms that result from an expansion of the latter 
product are all zero. This follows from the facts that, for 
$i\not=j$, $e\sb{i}e\sb{j}=e\sb{i}f\sb{j}=f\sb{i}e\sb{j}=
f\sb{i}f\sb{j}=0$ together with the data in Lemma \ref
{lm:prep1} (ii). For instance, $x\sb{i}x\sb{j}=x\sb{i}e
\sb{i}f\sb{j}x\sb{j}=0$.

As in the argument above, $p\sb{j}e\sb{i}=0$ for $i,j\in I$ with 
$j\not=i$, and it follows from Lemma \ref{lm:prep2} with $(q\sb{i})
\sb{i\in I}=(e\sb{i})\sb{i\in I}$ and $r=p\sb{i}$ that $ep\sb{i}=
e\sb{i}p\sb{i}$ and $p\sb{i}e=p\sb{i}e\sb{i}$ for all $i\in I$. 
Likewise, by  Lemma \ref{lm:prep2}, this time with $(q\sb{i})
\sb{i\in I}=(p\sb{i})\sb{i\in I}$ and $r=e\sb{i}$ we have $pe\sb{i}
=p\sb{i}e\sb{i}$ and $e\sb{i}p=e\sb{i}p\sb{i}$ for all $i\in I$.

By Lemma \ref{lm:prep1} (iii), $2e\sb{i}p\sb{i}e\sb{i}=e\sb{i}$, 
whence $2epe\sb{i}=2ep\sb{i}e\sb{i}=2e\sb{i}p\sb{i}e\sb{i}=e\sb{i}$, 
and we have $(2ep-1)e\sb{i}=0$, whereupon $(2ep-1)\dg e\sb{i}=0$ for 
all $i\in I$. Therefore, $e\sb{i}\leq((2ep-1)\dg)\sp{\perp}$ for 
all $i\in I$, and it follows that $e\leq((2ep-1)\dg)\sp{\perp}$, 
whence $(2ep-1)\dg e=0$, and consequently, $(2ep-1)e=0$, i.e., 
$2epe=e$. By similar arguments, $2pep=p$, $2fpf=f$, and $2pfp=p$. 

Let us write $h=ses=(2p-1)e(2p-1)=4pep-2ep-2pe+e$, noting that, 
since $s$ is a  symmetry, $h$ is a projection. Using the facts 
that $ef=fe=0$, $2pep=p$, and $2fpf=f$ we find that 
\[
fhf=f(4pep-2ep-2pe+e)f=4fpepf=2fpf=f. 
\]
Similarly using the facts that $ef=fe=0$, $2pfp=p$, and 
$2epe=e$, 
\[
hfh=(2p-1)e(2p-1)f(2p-1)e(2p-1)=(2p-1)(2ep-e)f(2pe-e)(2p-1)
\]
\[
=(2p-1)(4epfpe)(2p-1)=(2p-1)(2epe)(2p-1)=(2p-1)e(2p-1)=h. 
\]
Therefore $f(1-h)f=f-fhf=f-f=0$, whence $f(1-h)=0$, so $f\leq h$.
Likewise, since $hfh=h$, it follows that $h\leq f$, and we have 
$h=f$.
\end{proof}

\section{Central orthocompleteness} \label{sc:centOC} 

If $P$ is regarded as an effect algebra (Remark \ref
{rm:OMLeffectalg}) then the following definition of central 
orthocompleteness for the OML $P$ is equivalent to the 
effect-algebra definition of central orthocompleteness 
\cite[Definition 6.1]{COEA}. 

\begin{definition}
A family $(p\sb{i})\sb{i\in I}$ in the OML $P$ is \emph{centrally 
orthogonal} iff there is a pairwise orthogonal family $(c\sb{i})
\sb{i\in I}$ in the center $P\cap C(A)$ of $P$ such that 
$p\sb{i}\leq c\sb{i}$ for all $i\in I$. The projection lattice 
$P$ is \emph{centrally orthocomplete} iff every centrally 
orthogonal family $(p\sb{i})\sb{i\in I}$ in $P$ has a supremum 
$p=\bigvee\sb{i\in I}p\sb{i}$ in $P$.
\end{definition}

\noindent Obviously, if $P$ is complete as a lattice, then 
it is centrally orthocomplete. 

\begin{assumption} For the remainder of this article, we assume 
that the OML $P$ of projections in $A$ is centrally orthocomplete.
\end{assumption}

\begin{theorem} [{\cite[Theorem 6.8 (i)]{COEA}}] \label
{th:centercompl}
The center $P\cap C(A)$ of $P$ is a complete boolean algebra.
\end{theorem}

\begin{lemma} \label{lm:OCinterval}
Let $d\in P\cap C(A)$ and let $c$ be a projection in the 
synaptic algebra $dAd=dA=Ad$. Then{\rm: (i)} $c$ is a central 
projection in $dA$ iff $c$ is a central projection in $A$. 
{\rm(ii)} The OML $P[0,d]$ of projections in $dA$ is 
centrally orthocomplete.
\end{lemma}

\begin{proof} (i) Suppose that $c$ is a central projection 
in $dA$ and let $a\in A$. Then, as $c\leq d$, we have $c=
cd=dc$ and $ca=cda=dac=ac$, so $c\in P\cap C(A)$. The 
converse is obvious, and (i) is proved. 
\end{proof} 

\begin{lemma} [{\cite[Theorem 6.8 (ii)]{COEA}}] \label
{lm:centcover}
For each $p\in P$, there is a smallest central projection 
$c\in P\cap C(A)$ such that $p\leq c$.
\end{lemma}

\begin{definition} \label{df:centcover}  
If $a\in A$, then the smallest central projection 
$c\in P\cap C(A)$ such that $a\dg\leq c$ is called the 
\emph{central cover} of $a$ and denoted by $\gamma(a)$. 
The mapping $\gamma\colon A\to P\cap C{P}$ is called the 
\emph{central cover mapping}.
\end{definition}

Since $a\dg$ is the smallest projection $p\in P$ such that 
$ap=a$, it follows that $\gamma a$ is the smallest central 
projection $c\in P\cap C(A)$ such that $ac=a$. Moreover, 
by \cite[Theorems 5.2 and 6.10]{COEA}, the central cover 
mapping $\gamma$ has the following properties.

\begin{theorem} \label{th:centcovprops}
Let $p,q\in P$. Then{\rm:}
\begin{enumerate}
\item $\gamma 1=1$, $\gamma p=0\Leftrightarrow p=0$, and 
 $\gamma(P):=\{\gamma p:p\in P\}=P\cap C(A)$.
\item $\gamma(\gamma p)=\gamma p$ and $p\leq q\Rightarrow
 \gamma p\leq\gamma q$.
\item $\gamma(p\wedge\gamma q)=\gamma p\wedge\gamma q$.
\item $\gamma p\perp q\Leftrightarrow\gamma p\perp\gamma q
 \Leftrightarrow p\perp\gamma q\Rightarrow p\perp q$.
\item If $(p\sb{i})\sb{i\in I}$ is a family of elements in $P$ 
 and the supremum $\bigvee\sb{i\in I}p\sb{i}$ exists in $P$, then 
 $\bigvee\sb{i\in I}\gamma p\sb{i}$ exists in $P$ and 
 $\gamma(\bigvee\sb{i\in I}p\sb{i})=\bigvee\sb{i\in I}\gamma p\sb{i}$.
\end{enumerate}
\end{theorem}

\begin{theorem} \label{th:sup/infincenter}
{\rm(i)} The center $P\cap C(A)$ of $P$ is sup/inf-closed in $P$. 
{\rm(ii)} Let $(c\sb{i})\sb{i\in I}$ be a family of elements in the 
center $P\cap C(A)$ of $P$. Since $P\cap C(A)$ is a complete boolean 
algebra, the supremum $p$ and the infimum $q$ of $(c\sb{i})\sb{i\in I}$ 
exist in $P\cap C(A)$; moreover, $p$ and $q$ are, respectively, the 
supremum and the infimum of $(c\sb{i})\sb{i\in I}$ in $P$.
\end{theorem} 

\begin{proof}
Part (i) follows from the fact that the center of any OML is 
sup/inf-closed in the OML. Using the central cover mapping, one 
proceeds as in the proof of \cite[Theorem 5.2 (xiii)]{COEA} to 
prove part (ii).
\end{proof}

\section{Equivalence of projections} \label{sc:equivalenceproj} 

\emph{The assumption that $P$ is centrally orthocomplete is still 
in force.} In the next definition, we denote the composition of symmetry 
transformations by juxtaposition.

\begin{definition}
Let ${\mathcal J}$ be the set of all mappings $J\colon A\to A$ of 
the form $J=J\sb{s\sb{n}}J\sb{s\sb{n-1}}\cdots J\sb{s\sb{1}}$ where 
$s\sb{1},s\sb{2},...,s\sb{n}$ are symmetries in $A$. 
\end{definition}

\noindent Thus, ${\mathcal J}$ is the group under composition 
generated by the symmetry transformations on $A$. As a consequence 
of Theorem \ref{th:symmetry&quadratic}, the transformations $J\in 
{\mathcal J}$ have the following properties.  

\begin{theorem} \label{th:Jprops}
Let $J\in{\mathcal J}$, $a,b\in A$, and $e,f\in P$. Then{\rm:} 
{\rm(i)} The $J$ is an order, linear, and Jordan automorphism 
of $A$. {\rm(ii)} $J$ restricted to $P$ is an OML-automorphism 
of $P$. {\rm(iii)} If $ab=ba$, then $J(ab)=J(a)J(b)=J(b)J(a)
\in A$. {\rm(iv)} If $J(a)J(b)=J(b)J(a)$, then $ab=ba\in A$. 
{\rm(v)} $(J(a))\dg=J(a\dg)$. 
\end{theorem}

\begin{definition} \label{df:equivalence} 
Let $p,q\in P$. (i) The projections $p$ and $q$ are 
${\mathcal J}$-\emph{equivalent}, in symbols $p\sim q$, iff 
there exists $J\in{\mathcal J}$ such that $J(p)=q$. If 
$p\sim q$ and ${\mathcal J}$-equivalence is understood, we 
may simply say that $p$ and $q$ are \emph{equivalent}. (ii) The 
projections $p$ and $q$ are \emph{related} iff there are nonzero 
projections $p\sb{1}\leq p$ and $q\sb{1}\leq q$ such that $p\sb{1}
\sim q\sb{1}$. If $p$ and $q$ are not related, we say that they 
are \emph{unrelated}. (iii) $p$ is \emph{invariant} iff it is 
unrelated to its orthocomplement $p\sp{\perp}$. (iv) If there 
exists a projection $q\sb{1}\leq q$ such that $p\sim q\sb{1}$, 
we say that $p$ is \emph{sub-equivalent} to $q$, in symbols, 
$p\preceq q$.
\end{definition}

\noindent We note that $\sim$ is the transitive closure of the 
relation of being exchangeable by a symmetry and that, as a 
consequence of Theorems \ref{th:exstrongpersp} and \ref{th:perspex}, 
two projections $p$ and $q$ are equivalent iff they are 
projective in the OML $P$. 

Now we investigate the extent to which the equivalence relation 
$\sim$ is a \emph{Sherstnev-Kalinin} {\rm(}SK-{\rm)} 
\emph{congruence} on the OML $P$ \cite[\S 7]{HandD}. By 
definition, an SK-congruence satisfies axioms SK1--SK4 in 
\cite[Definition 7.2 and Remarks 7.3]{HandD}.  For these axioms 
we have: 

\smallskip 

\bul (SK1) Obviously, \emph{if $e\in P$ then} $e\sim 0\Rightarrow e=0$.

\smallskip

\bul (SK2) Axiom SK2, complete additivity (and even finite 
additivity) of $\sim$, is problematic. Theorem \ref{th:weakadditivity} 
which assumes completeness of the OML $P$, is a weak substitute 
for axiom SK2 and Lemma \ref{le:finadit} is a weak substitute 
for finite additivity.

\smallskip

\bul (SK3d) Axiom SK3d (divisibility) holds, in fact we have 
the following \emph{complete divisibility} property \cite[p.4]
{Loom}: \emph{If $(e\sb{i})\sb{i\in I}$ is an orthogonal family in 
$P$, $p\in P$, and $p\sim\bigvee\sb{i\in I}e\sb{i}$, then there 
exists an orthogonal family $(p\sb{i})\sb{i\in I}$ such that 
$p=\bigvee\sb{i\in I}p\sb{i}$ and $p\sb{i}\sim e\sb{i}$ for all 
$i\in I$.} Indeed, if $J\in{\mathcal J}$ with $p=J(\bigvee\sb
{i\in I}e\sb{i})=\bigvee\sb{i\in I}J(e\sb{i})$, then $p\sb{i}:=
J(e\sb{i})\sim e\sb{i}$ for all $i\in I$.

\smallskip

\bul (SK3e) Combining Theorems \ref{th:pqef} and \ref{th:perspex} 
(i), we find that $\sim$ satisfies axiom SK3e: \emph{If $p,q,e,f
\in P$, $p\perp q$, $e\perp f$, and $p\vee q=e\vee f$, then 
there exist $p\sb{1}, p\sb{2}, q\sb{1}, q\sb{2}\in P$ such that 
$p\sb{1}\perp p\sb{2},\,q\sb{1}\perp q\sb{2},\, p\sb{1}\perp q
\sb{1},\,p\sb{2}\perp q\sb{2},\,p\sb{1}\vee p\sb{2}=p,\,q\sb{1}
\vee q\sb{2}=q,\,p\sb{1}\vee q\sb{1}\sim e$, and $p\sb{2}\vee q
\sb{2}\sim f$.}  

\smallskip

\bul (SK4) \emph{Non-orthogonal projections are related, in fact, 
they have nonzero subprojections that are exchanged by a symmetry} 
(Theorem \ref{th:Sym&Sasaki} (iv)).

\smallskip

\noindent As we shall see, in spite of the fact that $\sim$ may 
not qualify as an SK-congruence, it does enjoy a number of 
important properties.

\begin{theorem}\label{th:key} 
Let $0\not=e,f\in P$ with $e\sim f$. Then $e$ and $f$ have 
nonzero subprojections that are exchanged by a symmetry. 
\end{theorem}

\begin{proof}
As $e\sim f$, there are symmetries $s\sb{1},s\sb{2},...,s
\sb{n}\in A$ such that 
\[
f=s\sb{n}s\sb{n-1}\cdots s\sb{2}s\sb{1}es\sb{1}s\sb{2}\cdots s
\sb{n-1}s\sb{n}.
\] 
\noindent The proof is by induction on $n$. For $n=1$ the 
desired conclusion is obvious. Suppose that the conclusion holds 
for all sequences of symmetries of length $n-1$ and let $r:=s
\sb{n}fs\sb{n}=s\sb{n-1}\cdots s\sb{2}s\sb{1}es\sb{1}s\sb{2}
\cdots s\sb{n-1}$. By the induction hypothesis, there are nonzero 
subprojections $p\leq e$ and $q\leq r$ and a symmetry $s\in S$ 
such that $sps=q$. Let $k:=s\sb{n}qs\sb{n}\leq s\sb{n}rs
\sb{n}=f$. If $p\not\perp k$, then by SK4, there are nonzero 
subprojections $p\sb{1}\leq p\leq e$ and $k\sb{1}\leq k\leq f$ 
that are exchanged by a symmetry, and our proof is complete. Thus, 
we can and do assume that $p\perp k$. Therefore, since  $k=s
\sb{n}spss\sb{n}$, it follows from Theorem \ref{th:perspex} (ii) 
that $p$ and $k$ are exchanged by a symmetry.
\end{proof}

In \cite[\S 7]{HandD}, for an SK-congruence, the infimum of all 
the invariant elements that dominate an element is called the 
\emph{hull} of that element. Thus, by the following theorem, 
the central cover mapping is an analogue for the equivalence 
relation $\sim$ of a hull mapping for an SK-congruence.

\begin{theorem}\label{th:invarcentral} 
Let $h\in P$. Then the following conditions are mutually 
equivalent{\rm: (i)} $h$ is invariant. {\rm(ii)} If $p\in P$, 
$s$ is a symmetry in $A$, and $sps\leq h$, then $p\leq h$. 
{\rm(iii)} If $p\in P$ and $p\preceq h$, then $p\leq h$. 
{\rm(iv)} If $q\in P$, then $q\wedge h=0\Rightarrow q\perp h$. 
{\rm(v)} $h$ is central.
\end{theorem}

\begin{proof}
(i) $\Rightarrow$ (ii). Let $h$ be invariant $p\in P$, and 
let $s\in A$ be a symmetry such that $h\sb{1}:=sps\leq h$. 
Aiming for a contradiction, we assume that $p$ is related to 
$h\sp{\perp}$, i.e., there exist subprojections $0\not=p
\sb{1}\leq p$ and $0\not=q\sb{1}\leq h\sp{\perp}$ with $p
\sb{1}\sim q\sb{1}$. Now $h\sb{1}=sps=s(p\sb{1}\vee(p
\wedge p\sb{1}\sp{\,\perp}))s=sp\sb{1}s\vee s(p\wedge p
\sb{1}\sp{\,\perp})s $, whence $0\not=h\sb{2}:=sp\sb{1}s\leq h
\sb{1}\leq h$. But then $h\geq h\sb{2}\sim p\sb{1}\sim q
\sb{1}\leq h\sp{\perp}$, so $h$ is related to $h\sp{\perp}$, 
contradicting the invariance of $h$. Therefore, $p$ is 
unrelated to $h\sp{\perp}$ and it follows from SK4 that 
$p\perp h\sp{\perp}$, i.e., $p\leq h$.

(ii) $\Rightarrow$ (iii). If $p\preceq h$, there are 
symmetries $s\sb{1},s\sb{2},...s\sb{n}\in A$ such that 
$s\sb{n}s\sb{n-1}\cdots s\sb{1}ps\sb{1}s\sb{2}\cdots s\sb{n}
\leq h$, whence (iii) follows from (ii) by induction on 
$n$.
 
(iii) $\Rightarrow$ (iv). Assume that (iii) holds and that 
$q\not\perp h$. By SK4, there exist subprojections $0\not=
q\sb{1}\leq q$ and $0\not=h\sb{1}\leq h$ with $q\sb{1}\sim 
h\sb{1}$. Then $q\sb{1}\preceq h$, so $q\sb{1}\leq h$ by 
(iii), and we have $0\not=q\sb{1}\leq q\wedge h$, whence 
$q\wedge h\not=0$.  

(iv) $\Rightarrow$ (v). Assume that (iv) holds and let 
$q\in P$. Then $\phi\sb{q\sp{\perp}}(h\sp{\perp})\wedge h
=(q\vee h\sp{\perp})\wedge q\sp{\perp}\wedge h=0$, so 
$\phi\sb{q\sp{\perp}}(h\sp{\perp})\leq h\sp{\perp}$ by 
(iii). Therefore, $q\sp{\perp}$ is compatible with $h\sp{\perp}$, 
and it follows that $q$ is compatible with $h$. Since $q$  
is an arbitrary element of $P$, it follows that $h$ is in the 
center of $P$.

(iv) $\Rightarrow$ (i). Assume that $h\in P\cap C(A)$ and, 
aiming for a contradiction, suppose that $h$ is related to 
$h\sp{\perp}$. Then there exist subprojections $0\not=h\sb{1}
\leq h$ and $0\leq q\sb{1}\leq h\sp{\perp}$ such that 
$h\sb{1}\sim q\sb{1}$. Thus there are symmetries $s\sb{1}, 
s\sb{2},...,s\sb{n}\in A$ such that $q\sb{1}=s\sb{n}s\sb{n-1}
\cdots s\sb{1}h\sb{1}s\sb{1}\cdots s\sb{n-1}s\sb{n}$. As $h\sb{1}
\leq h$, we have $s\sb{1}h\sb{1}s\sb{1}\leq s\sb{1}hs\sb{1}=
hs\sb{1}\sp{\,2}=h$, and by induction on $n$, $q\sb{1}\leq h$.
Consequently, $q\sb{1}\leq h\wedge h\sp{\perp}=0$, contradicting 
$q\sb{1}\not=0$.
\end{proof}

\begin{corollary} \label{co:punrelc}
If $p\in P$ and $c\in P\cap C(A)$, then $p$ and $c$ are unrelated 
iff $p\perp c$.
\end{corollary}

\begin{proof}
If $p$ and $c$ are unrelated, then $p\perp c$ by SK4. Conversely, 
suppose $p\perp c$, $p\sb{1}\leq p$, $c\sb{1}\leq c$ and $p\sb{1}
\sim c\sb{1}$. Then $p\sb{1}\preceq c$, so $p\sb{1}\leq c$ by 
Theorem \ref{th:invarcentral}. But $p\sb{1}\leq p\leq c\sp{\perp}$, 
so $p\sb{1}\leq c\wedge c\sp{\perp}=0$; hence $p$ and $c$ are 
unrelated. 
\end{proof}

\begin{theorem}\label{th:centcovsup} 
Suppose that $P$ is a complete OML and $p\in P$. Then \newline  
$\gamma p=\bigvee\{q\in P; q\preceq p\}$.
\end{theorem}

\begin{proof}
Since $P$ is complete, $h:=\bigvee\{q\in P:q\preceq p\}$ 
exists in $P$. If $q\preceq p$, then since $p\leq\gamma p$, we 
have $q\preceq\gamma p\in P\cap C(A)$, whence $q\leq\gamma p$ by 
Theorem \ref{th:invarcentral}. Therefore $h\leq\gamma p$. We 
claim that $h$ is a central projection. Indeed, suppose $r\in P$, 
$s\in A$ is a symmetry, and $srs\leq h$. By Theorem \ref
{th:invarcentral} it will be sufficient to show that $r\leq h$. 
But $srs\leq h$ implies that $r\leq shs=\bigvee\{sqs:q\preceq p\}$, 
and since $q\preceq p$ implies $sqs\preceq p$, it follows that 
$\bigvee\{sqs:q\preceq p\}\leq h$. Therefore, $r\leq h$, whence 
$h\in P\cap C(A)$. Obviously, $p\leq h$, so $\gamma p\leq \gamma h
=h$, and we have $h=\gamma p$.
\end{proof}

\begin{corollary}\label{co:gammaperpunrel} 
Let $P$ be a complete OML. Then{\rm: (i)} If $e,f\in P$, then 
$\gamma e\perp\gamma f$ iff $e\perp\gamma f$ iff $e$ and $f$ are 
unrelated. {\rm(ii)} $P$ is irreducible, i.e., $P\cap C(A)
=\{0,1\}$, iff every pair of nonzero projections are related. 
\end{corollary}

\begin{proof} (i) If $\gamma e\perp \gamma f$, than as $e\leq 
\gamma e$, it follows that $e\perp\gamma f$. If $e\perp\gamma f$, 
then since $\gamma f\in P\cap C(A)$, Corollary \ref{co:punrelc} 
implies that $e$ and $\gamma f$ are unrelated, and since $f\leq
\gamma f$, it follows that $e$ is unrelated to $f$. To complete 
the proof of (i) it will be sufficient to prove that $\gamma e
\not\perp\gamma f$ implies that $e$ is related to $f$. So 
assume that $\gamma e\not\leq(\gamma f)\sp{\perp}$ Then $e\not
\leq(\gamma f)\sp{\perp}$, else $\gamma e\leq\gamma(\gamma f)
\sp{\perp}=(\gamma f)\sp{\perp}$. Thus by Theorem \ref{th:centcovsup} 
and De\,Morgan duality, $e\not\leq(\gamma f)\sp{\perp}=\bigwedge
\{q\sp{\perp}:q\preceq f\}$, whence there exists a projection 
$q\preceq f$ such that $e\not\perp q$. Therefore, by SK4, $e$ is 
related to $q$, and since $q\preceq f$, it follows that $e$ is related 
to $f$.

(ii) Part (ii) follows immediately from (i).
\end{proof}

By \cite[Definition 7.14]{HandD}, an SK-congruence is a \emph{dimension 
equivalence relation} (DER) iff unrelated elements have orthogonal 
hulls. Therefore, by Corollary \ref{co:gammaperpunrel} (i), if $P$ 
is complete, then the equivalence relation $\sim$ is an analogue of 
a DER. 

\section{The case of a complete projection lattice} \label{sc:CompleteP} 

Some of the results above, notably Theorems \ref{th:weakadditivity} 
and \ref{th:centcovsup} require the completeness of the OML $P$. 
In this section, we present some additional results involving 
exchangeability of projections by symmetries that also require 
completeness of $P$.

\begin{assumption}
In this section, we assume that $P$ is a complete lattice. 
Therefore, $P$ is centrally orthocomplete, and the central cover 
mapping $\gamma\colon A\to P\cap C(A)$ exists. 
\end{assumption}

\begin{theorem}
Suppose that $p\in P$, and $S$ is the set of all symmetries 
in $A$. Then $\gamma p=\bigvee\{sqs:s\in S,\, q\in P,\text
{\ and\ }q\leq p\}$.
\end{theorem}

\begin{proof}
Let $h:=\bigvee\{sqs:q\in P\text{\ and\ }q\leq p\}$. If 
$s\in S$, $q\in P$, and $q\leq p$, then $q\leq\gamma p$, 
whence $sqs\leq s(\gamma p)s=\gamma p$, and therefore 
$h\leq\gamma p$. Aiming for a contradiction, we assume 
that $h\not=\gamma p$, i.e., that $r:=\gamma p-h=\gamma p
\wedge h\sp{\perp}\not=0$. Then $r\not\leq(\gamma p)\sp
{\perp}$ so by Corollary \ref{co:gammaperpunrel}, $r$ and 
$p$ are related. Therefore, there exist $0\not=r\sb{1}\leq r$ 
and $0\not=p\sb{1}\leq p$ with $r\sb{1}\sim p\sb{1}$; hence, 
by Theorem \ref{th:key} there exist $0\not=r\sb{2}\leq r\sb{1}
\leq r\leq h\sp{\perp}$ and $0\not=p\sb{2}\leq p\sb{1}\leq p$ 
such that $r\sb{2}$ and $p\sb{2}$ are exchanged by a symmetry     
$s\in A$. Thus, $r\sb{2}=sp\sb{2}s\leq h$, so $r\sb{2}\leq h
\wedge h\sp{\perp}=0$, contradicting $r\sb{2}\not=0$. 
\end{proof}

\begin{lemma}\label{lm:orthodecomps} 
If $e,f\in P$ are orthogonal projections, then there exist 
projections $e\sb{1}, e\sb{2}, f\sb{1}, f\sb{2}\in P$ 
such that $e\sb{1}\perp e\sb{2}$, $f\sb{1}\perp f\sb{2}$, 
$e=e\sb{1}\vee e\sb{2}=e\sb{1}+e\sb{2}$, $f=f\sb{1}\vee f
\sb{2}=f\sb{1}+f\sb{2}$, $e\sb{1}$ and $f\sb{1}$ are exchanged 
by a symmetry, and $e\sb{2}$ is unrelated to $f\sb{2}$, whence 
$\gamma e\sb{2}\perp\gamma f\sb{2}$.
\end{lemma}

\begin{proof} Let $(e\sb{i},f\sb{i})\sb{i\in I}$ be a maximal 
family of pairs of projections such that $(e\sb{i})\sb{i\in I}$ 
is an orthogonal family of subprojections of $e$, $(f\sb{i})
\sb{i\in I}$ is an orthogonal family of subprojections of $f$, 
and for each $i\in I$, there is a symmetry $s\sb{i}\in A$ that 
exchanges $e\sb{i}$ and $f\sb{i}$. We can assume that the natural 
numbers $1,2,3$ and $4$ do not belong to the indexing set $I$. 

Put $e\sb{1}:=\bigvee\sb{i\in I}e\sb{i}$ and $f\sb{1}:=\bigvee
\sb{i\in I}f\sb{i}$. Then $e\sb{1}\leq e$ and $f\sb{1}\leq f$, 
so $e\sb{1}\perp f\sb{1}$. By Theorem \ref{th:weakadditivity}, 
$e\sb{1}$ and $f\sb{1}$ are exchanged by a symmetry. Let $e\sb{2}
:=e-e\sb{1}$ and $f\sb{2}:=f-f\sb{1}$. Then $e\sb{2}\leq e\wedge
e\sb{i}\sp{\,\perp}$ and $f\sb{2}\leq f\wedge f\sb{i}\sp{\,\perp}$ 
for all $i\in I$. Suppose that $e\sb{2}$ is  related to $f\sb{2}$; 
then they have nonzero subprojections $0\not=e\sb{3}\leq e\sb{2}$ 
and $0\not=f\sb{3}\leq f\sb{2}$ with $e\sb{3}\sim f\sb{3}$. Thus, 
by Theorem \ref{th:key}, there are nonzero subprojections $0\not=e
\sb{4}\leq e\sb{3}\leq e\sb{2}$ and $0\not=f\sb{4}\leq f\sb{3}
\leq f\sb{2}$ that are exchanged by a symmetry $s\sb{4}\in A$. 
Evidently, $e\sb{4}\leq e\wedge e\sb{i}\sp{\,\perp}$ and $f\sb{4}
\leq f\wedge f\sb{i}\sp{\,\perp}$ for all $i\in I$. But then we 
can enlarge the family $(e\sb{i},f\sb{i})\sb{i\in I}$ by appending 
the pair $(e\sb{4},f\sb{4})$, contradicting maximality. Therefore 
$e\sb{2}$ is unrelated to $f\sb{2}$, whence $\gamma e\sb{2}\perp 
\gamma f\sb{2}$ by Corollary \ref{co:gammaperpunrel}.
\end{proof}

In the next theorem we improve Lemma \ref{lm:orthodecomps}, by 
removing the hypothesis that $e$ and $f$ are orthogonal.

\begin{theorem}\label{th:orthdecomps} 
If $e$ and $f$ are any two projections, then we can write orthogonal 
sums $e=e\sb{1}+e\sb{2}$ and $f=f\sb{1}+f\sb{2}$, where $e\sb{1}$ 
and $f\sb{1}$ are exchanged by a symmetry and $\gamma e\sb{2}\perp
\gamma f\sb{2}$.
\end{theorem}

\begin{proof} Write $e\sb{11}:=\phi\sb{e}(f)=e\wedge(e\sb{12})
\sp{\perp}$, where $e\sb{12}:=e\wedge f\sp{\perp}$, and write 
$f\sb{11}:=\phi\sb{f}(e)=f\wedge(f\sb{12})\sp{\perp}$, where 
$f\sb{12}:=e\sp{\perp}\wedge f$. Then $e\sb{12}$ commutes with 
both $e$ and $(e\sb{12})\sp{\,\perp}$, whence $e=e\sb{11}\vee e
\sb{12}=e\sb{11}+e\sb{12}$, and likewise $f=f\sb{11}+f\sb{12}$. 
Also by Theorem \ref{th:Sym&Sasaki} (i) $e\sb{11}$ and $f\sb{11}$ 
are exchanged by a symmetry $s\sb{1}$. Moreover, since $e\sb{12}
\perp f\sb{12}$, Lemma \ref{lm:orthodecomps} provides orthogonal 
decompositions $e\sb{12}=e\sb{13}+e\sb{2}$ and $f\sb{12}=f\sb{13}
+f\sb{2}$ where $e\sb{13}$ and $f\sb{13}$ are exchanged by a 
symmetry $s\sb{2}$ and $\gamma e\sb{2}\perp\gamma f\sb{2}$. Thus 
$e=e\sb{11}+e\sb{12}=e\sb{11}+e\sb{13}+e\sb{2}$ and $f=f\sb{11}+
f\sb{12}=f\sb{11}+f\sb{13}+f\sb{2}$.

Since $e\sb{13}\perp f\sb{11}$ and $e\sb{11}\perp f\sb{13}$, Lemma 
\ref{le:finadit} provides a symmetry $s$ exchanging $e\sb{1}:=e\sb{11}
+e\sb{13}$ and $f\sb{1}:=f\sb{11}+f\sb{13}$.
\end{proof}

\begin{lemma}\label{le:orthogcompar} 
If $e$ and $f$ are orthogonal projections, then there 
is a central projection $h\in P$ such that $eh$ and a 
subprojection of $fh$ are exchanged by a symmetry $s$, 
and $f(1-h)$ and a subprojection of $e(1-h)$ are exchanged 
by the symmetry $s$.
\end{lemma}

\begin{proof} Let $e=e\sb{1}+e\sb{2}$, $f=f\sb{1}+f\sb{2}$ be 
the decompositions of Theorem \ref{th:orthdecomps} and set 
$h=\gamma f\sb{2}$. Then $h$ is a central projection, $f
\sb{2}h=f\sb{2}$, $h\perp\gamma e\sb{2}$, and there is a symmetry 
$s\in A$ with $se\sb{1}s=f\sb{1}$. Thus, $eh=e\sb{1}h+e\sb{2}h=
e\sb{1}h+e\sb{2}\gamma e\sb{2}h=e\sb{1}h$, and $s(eh)s=s(e\sb{1}h)s
=f\sb{1}h\leq fh$. Also $f(1-h)=f\sb{1}(1-h)+f\sb{2}(1-h)=f\sb{1}
(1-h)$ and $sf(1-h)s=sf\sb{1}(1-h)s=e\sb{1}(1-h)\leq e(1-h)$.
\end{proof}

In the next theorem we improve Lemma \ref{le:orthogcompar} by 
removing the hypothesis that $e$ and $f$ are orthogonal.

\begin{theorem} [Generalized Comparability]  \label{th:compar} 
Given any two projections $e,f$ there is a central projection 
$h$ and a symmetry $s$ with $s(eh)s\leq fh$ and $sf(1-h)s\leq e(1-h)$.
\end{theorem}

\begin{proof} 
By Theorem \ref{th:Sym&Sasaki} (i) the subprojections $e\sb{1}:
=\phi\sb{e}(f)\leq e$ and $f\sb{1}:=\phi\sb{f}(e)\leq f$ are 
exchanged by a symmetry $s\sb{1}$. Set $e\sb{2}=e\wedge f
\sp{\perp}$ and $f\sb{2}=e\sp{\perp}\wedge f$. Then $e\sb{1}
=e\wedge e\sb{2}\sp{\,\perp}$, $f\sb{1}=f\wedge f\sb{2}\sp
{\,\perp}$, $e\sb{1}\perp e\sb{2}$, $f\sb{1}\perp f\sb{2}$, 
$e=e\sb{1}+e\sb{2}$, and $f=f\sb{1}+f\sb{2}$. Since $e\sb{2}
\perp f\sb{2}$, Lemma \ref{le:orthogcompar} applies giving a 
symmetry $s\sb{2}$ and a central projection $h$ with $f\sb{3}
:=s\sb{2}(e\sb{2}h)s\sb{2}\leq f\sb{2}h$ and $e\sb{3}:=s\sb{2}
(f\sb{2}(1-h))s\sb{2}\leq e\sb{2}(1-h)$. We note that $f\sb{3}
(1-h)=0$ and $e\sb{3}h=0$. 

We claim that the projections $e\sb{1}$, $e\sb{2}h$ and 
$e\sb{3}$ are pairwise orthogonal. Indeed, as $e\sb{1}
\perp e\sb{2}$, we have $e\sb{1}\perp e\sb{2}h$. Also, 
$e\sb{1}\perp e\sb{2}(1-h)$, so $e\sb{1}\perp e\sb{3}$.  
Moreover, $e\sb{3}\leq 1-h$, so $e\sb{2}h\perp e\sb{3}$. 
Thus, $e\sb{1}+(e\sb{2}h+e\sb{3})$ is a projection. 
Similarly, $f\sb{1}+(f\sb{3}+f\sb{2}(1-h))$ is a projection.
Since $e\sb{1}\leq e\leq e\vee f\sp{\perp}=f\sb{2}\sp{\,\perp}$  
and $f\sb{3}\leq f\sb{2}$ it follows that $e\sb{1}\perp(f\sb{3}+
(1-h)f\sb{2})$. Similarly, $f\sb{1}\perp(e\sb{2}h+e\sb{3})$, 
Thus, as $s\sb{1}$ exchanges $e\sb{1}$ and $f\sb{1}$ and 
$s\sb{2}$ exchanges $e\sb{2}h+e\sb{3}$ and $f\sb{3}+
f\sb{2}(1-h)$, it follows from Lemma \ref{le:finadit} that 
there is a symmetry $s\in S$ such that 
\begin{equation} \label{eq:first}
s(e\sb{1}+(e\sb{2}h+e\sb{3}))s=f\sb{1}+(f\sb{3}+f\sb{2}(1-h)), 
\end{equation}
\begin{equation}\label{eq:second}
\text{\ whence\ }s(f\sb{1}+(f\sb{3}+f\sb{2}(1-h)))s=e\sb{1}+
 (e\sb{2}h+e\sb{3}).
\end{equation}
Multiplying both sides of (\ref{eq:first}) by $h$, and both sides of 
(\ref{eq:second}) by $1-h$, we find that $s(eh)s=(f\sb{1}+f\sb{3})h
\leq fh$ and $s(f(1-h))s=(e\sb{1}+e\sb{3})(1-h)\leq e(1-h)$.   
\end{proof}

For the case under consideration in which $P$ is a complete OML, 
the generalized comparability theorem above can be used to prove 
that $P$ has the relative center property. Our proof of the 
following theorem is suggested by the proof of \cite[Proposition 1]
{Chev} in which $\sim$ is replaced by strong perspectivity.    

\begin{theorem} \label{th:relativecenter}
The OML $P$ has the relative center property.
\end{theorem}

\begin{proof}
Let $p\in P$ and suppose that $d$ is a central element of 
the $p$-interval $P[0,p]$. Then $pAp$ is a synaptic algebra 
with unit $p$, $P[0,p]$ is the projection lattice of $pAp$, 
and $d$ commutes with every element of $pAp$.

Applying Theorem \ref{th:compar} to the projections $d$ and 
$p\wedge d\sp{\perp}=p-d$, we find that there is a symmetry 
$s\in A$ and central element $h\in P\cap C(A)$ such that $sdhs
\leq(p-d)h$ and $s(p-d)(1-h)s\leq d(1-h)$. Put $q:=dh\vee sdhs
\in P$. Then $sqs=sdhs\vee dh=q$, so $sq=qs$. Also, $dh\leq d
\leq p$, $sdhs\leq(p-d)h\leq p-d\leq p$, and we have $dh\leq q
\leq p$. Let $s\sb{0}:=sq=qs$ and $t:=s\sb{0}+(p-q)$. Then $s
\sb{0}\sp{\,2}=q$, $t\sp{2}=p$, and $sdhs=s\sb{0}dhs\sb{0}=tdht$.
Also, $ptp=t$, so $t\in pAp$, and it follows that $dt=td$; 
hence $dh=tdht=sdhs\leq (p-d)h\leq d\sp{\perp}$, and it follows 
that $dh=0$.  An analogous argument shows that $(p-d)(1-h)
=0$, and it follows that $d=d+(p-d)(1-h)=p(1-h)$.
\end{proof}


\begin{thebibliography}{99}

\bibitem{Alf} Alfsen, E.M., \emph{Compact Convex Sets and Boundary
Integrals}, Springer-Verlag, New York, 1971, ISBN 0-387-05090-6.

\bibitem{Beran} Beran, L., {\em Orthomodular Lattices, An Algebraic 
Approach}, Mathematics and its Applications, Vol. 18, D. Reidel Publishing 
Company, Dordrecht, 1985.

\bibitem{Chev} Chevalier, G., Around the relative center property in 
orthomodular lattices, \emph{Proc. Amer. Math. Soc.} {\bf 112} (1991), 
935--948.

\bibitem{OMLNote} Foulis, D.J., A note on orthomodular lattices, 
\emph{Portugal. Math.} {\bf 21} (1962) 65--72.

\bibitem{FCPOAG} Foulis, D.J., Compressions on partially ordered
abelian groups, \emph{Proc. Amer. Math. Soc.} {\bf 132} (2004)
3581--3587.

\bibitem{FSynap} Foulis, D.J., Synaptic algebras, \emph{Math. Slovaca} 
{\bf 60}, no. 5 (2010) 631--654.

\bibitem{FPMonotone} Foulis, D.J. and Pulmannov\'{a}, S., Monotone
sigma-complete RC-groups, \emph{J. London Math. Soc.} {\bf 73},
No. 2 (2006) 304--324.

\bibitem{SOUS} Foulis, D.J. and Pulmannov\'{a}, S., Spectral resolution
in an order unit space, \emph{Rep. Math. Phys.} {\bf 62}, no. 3 (2008) 
323--344.

\bibitem{GHAlg1} Foulis, D.J. and Pulmannov\'{a}, S., Generalized
Hermitian Algebras, \emph{Int. J. Theor. Phys.} {\bf 48}, no. 5 (2009) 
1320--1333.

\bibitem{FPSynap} Foulis, D.J. and Pulmannov\'{a}, S., Projections in 
synaptic algebras, \emph{Order} {\bf 27}, no. 2 (2010) 235--257.

\bibitem{COEA} Foulis, D.J. and Pulmannov\'{a}, S., Centrally 
orthocomplete effect algebras, \emph{Algebra Univers.} {\bf 64} 
(2010) 283--307, DOI 10.007/s00012-010-0100-5.

\bibitem{GHAlg2} Foulis, D.J. and Pulmannov\'{a}, S., Regular elements
in generalized Hermitian Algebras, \emph{Math. Slovaca} {\bf 61}, no. 2 
(2011) 155--172.

\bibitem{HandD} Foulis, D.J. and Pulmannov\'{a}, S., Hull mappings and 
dimension effect algebras, \emph{Math. Slovaca} {\bf 61}, no. 3 (2011) 
1--38. 

\bibitem{GPBB} Gudder, S., Pulmannov\'{a}, S., Bugajski, S., and
Beltrametti, E., Convex and linear effect algebras, \emph{Rep. Math.
Phys.} {\bf 44}, no. 3 (1999) 359--379.

\bibitem{SSH64} Holland, S.S, Jr., Distributivity and perspectivity 
in orthomodular lattices, \emph{Trans. Amer. Math. Soc.} {\bf 112} 
(1964) 330-343.

\bibitem{SSH70} Holland, S.S., Jr., An $m$-orthocomplete orthomodular
lattice is $m$-complete, \emph{Proc. Amer. Math. Soc.} {\bf 24} (1970)
716--718.

\bibitem{JPorthocompl} Jen\v{c}a, G. and Pulmannov\'{a}, S., 
Orthocomplete effect algebras, \emph{Proc. Amer. Math. Soc.} 
{\bf 131} (2003) 2663--2671.

\bibitem{Kalm} Kalmbach, G., \emph{Orthomodular Lattices}, Academic Press,
London, New York, 1983.

\bibitem{Kappro} Kaplansky, I., Projections in Banach algebras,
\emph{Annals of Mathematics} {\bf 53} no. 2, (1951) 235--249.

\bibitem{Kapcg} Kaplansky, I., Any orthocomplemented complete modular 
lattice is a continuous geometry, \emph{Ann. of Math.} {\bf 61} 
(1955) 524--541.

\bibitem{Loom} Loomis, L.H., The lattice theoretic background of the 
dimension theory of operator algebras, \emph{Memoirs of AMS} {\bf No 13} 
(1955) 1--36.

\bibitem{McC} McCrimmon, K.,  \emph{A taste of Jordan algebras},
Universitext, Springer-Verlag, New York, 2004, ISBN: 0-387-95447-3.

\bibitem{vNcg} von Neumann, J., \emph{Continuous geometry}, Princeton Univ. 
Press, Princeton 1960.

\bibitem{PtPu} Pt\'ak, P. and Pulmannov\'a, S., \emph{Orthomodular
Structures as Quantum Logics}, Kluwer Academic Publ., Dordrecht,
Boston, London, 1991.

\bibitem{PuNote} Pulmannov\'{a}, S., A note on ideals in synaptic algebras, 
\emph{Math. Slovaca} {\bf 62}, no. 6 (2012) 1091-–1104. 

\bibitem{Sikor} Sikorsky, R., \emph{Boolean Algebras}, 2$\sp{\rm nd}$ ed.,
Academic Press, New York and Springer Verlag, Berlin, 1964.

\end{thebibliography}
\end{document}